\documentclass{eptcs}
\providecommand{\event}{Draft} % Name of the event you are submitting to
\usepackage{breakurl}             % Not needed if you use pdflatex only.

\usepackage{amsfonts, amssymb, amscd, amsthm, amsmath, xspace}
\usepackage{IEEEtrantools}

\usepackage{color}
\usepackage{graphicx}
\usepackage{epstopdf}
\usepackage{subfigure}

\newtheorem{theorem}{Theorem}
\newtheorem{lemma}{Lemma}

\newtheorem{corollary}{Corollary}
\newtheorem{remark}{Remark}

\newtheorem{definition}{Definition}

\newcommand{\etal}{{\it et al.}}
\newcommand{\ie}{{\it i.e.}}
\newcommand{\eg}{{\it e.g.}}

\newcommand{\GF}[1]{\mathbb{F}_{#1}} % Galois field
\newcommand{\code}{\mathcal{C}} % Code C
\newcommand{\dual}{\mathcal{C}^{\perp}} % Code C
\newcommand{\vect}[1]{\mathbf{#1}} %vector 
\newcommand{\cw}{\mathbf{c}} %vector c
\newcommand{\supp}[1]{\textsf{Supp}\left(#1\right)} % support
\newcommand{\rnk}[1]{\textsf{Rank}\left(#1\right)} % rank
\newcommand{\rate}[1]{\textsf{Rate}\left(#1\right)} % rate
\newcommand{\wt}[1]{\textsf{wt}\left(#1\right)} % wt
\newcommand{\dims}[1]{\textsf{dim}\left(#1\right)} % dimension
\newcommand{\rep}[2]{\mathcal{R}_{#1}\left(#2\right)} % repair cw
\newcommand{\repcw}[2]{\tilde{\cw}_{#1,#2}} % repair c/w
\newcommand{\subspace}[1]{\langle #1 \rangle} % < . >

\newcommand{\set}[1]{\mathcal{#1}} % Set in cal font
\newcommand{\Nj}[1]{N_{#1}} % N_j

\newcommand{\mip}{m+i'} % m + i' : size of maximal pairwise disjoint 3-subsets
\newcommand{\xsubi}[1]{x_{#1}} % x_i
\newcommand{\Esubj}[1]{E_{#1}} % E_j
\newcommand{\Fsubj}[1]{F_{#1}} % F_j
\newcommand{\Tsubj}[1]{T_{#1}} % T_j
\newcommand{\A}{A} % A : union of support of maximal set of pairwise disjoint 3-subsets
\newcommand{\Ap}{A'} % A' = [n] - A
\newcommand{\App}{A''} % A''
\newcommand{\Aonep}{A'_1} % union of 3-subsets meeting A in 2 points in A'
\newcommand{\Bonep}{B'_1} % union of 3-subsets meeting A in 2 points in A'
\newcommand{\Conep}{C'_1} % union of 3-subsets meeting A in 2 points in A'

\newcommand{\codeD}{\mathcal{D}} % Code D
\newcommand{\codeDp}{\mathcal{D}'} % Code D'
\newcommand{\codeproj}[2]{{#1}\mid_{#2}} % Projection of a code on a set

\newcommand{\nkd}{(n,k,d)} % (n,k,d)
\newcommand{\nkdrt}{(n,k,d,r,t)} % (n,k,d,r,t)
\newcommand{\rt}{(r,t)} % (r,t)
\newcommand{\class}{\hat{\code}} %class of codes generated by exact number of covering codewords

\newcommand{\vl}{\mathbf{v}} %coset leader v
\newcommand{\w}{w} %weight of the leader w

\title{Rate Optimal Binary Linear Locally Repairable Codes with Small Availability}
\author{Swanand Kadhe and Robert Calderbank
\institute{}
\email{swanand.kadhe@tamu.edu, robert.calderbank@duke.edu}
}
\def\titlerunning{Binary LRCs with Availability}
\def\authorrunning{Kadhe and Calderbank}

\begin{document}
\maketitle

\begin{abstract}
A locally repairable code with availability has the property that every code symbol can be recovered from multiple, disjoint subsets of other symbols of small size. In particular, a code symbol is said to have $(r,t)$-\emph{availability} if it can be recovered from $t$ disjoint subsets, each of size at most $r$. A code with availability is said to be \emph{rate-optimal}, if its rate is maximum among the class of codes with given locality, availability, and alphabet size.

This paper focuses on rate-optimal binary, linear codes with small availability, and makes four  contributions. First, it establishes tight upper bounds on the rate of binary linear codes with $(r,2)$ and $(2,3)$ availability. Second, it establishes a uniqueness result for binary rate-optimal codes, showing that for certain classes of binary linear codes with $(r,2)$ and $(2,3)$-availability, any rate optimal code must be a direct sum of shorter rate optimal codes. Third, it presents novel upper bounds on the rates of binary linear codes with $(2,t)$ and $(r,3)$-availability. In particular, the main contribution here is a new method  for bounding the number of cosets of the dual of a code with availability, using its covering properties. Finally, it presents a class of locally repairable linear codes associated with convex polyhedra, focusing on the codes associated with the Platonic solids. It demonstrates that these codes are locally repairable with $t = 2$, and that the codes associated with (geometric) dual polyhedra are (coding theoretic) duals of each other. 
\end{abstract}

\section{Introduction}
\label{sec:intro}
The enormous growth of data being stored or computed online has encouraged practical distributed storage systems to migrate from triple replication~\cite{Rowstron:01, Ghemawat:03:GFS} to erasure coding for handling failures, see, \eg,~\cite{Huang:12, Muralidhar:14}. 
Even though classical erasure codes such as Reed-Solomon codes achieve high storage efficiency, they are inefficient in handling disk (or node) failures as they usually require to download large amount of data while repairing a failed node. The conflicting requirements of reliability, storage efficiency, and repair efficiency in data centers have created a new set of problems for coding theorists. Two measures of repair efficiency have received particular research attention: (a) {\it repair bandwidth} -- the metric is the total number of symbols (or bits) communicated while repairing a failed node, and the corresponding family of codes is called {\it regenerating codes} (see, \eg,~\cite{Dimakis:10, Dimakis:11,Rashmi:11}); and (b) {\it repair locality} -- the metric is the number of nodes participating in the repair process, and the corresponding family of codes is called  {\it locally repairable codes} (see, \eg,~\cite{Huang:07, Gopalan:12, Oggier:11, Sathiamoorthy:13, TamoB:14}). We restrict our attention to codes with locality in this work.

A locally repairable code (LRC) is a code of length $n$ over a finite field $\GF{}$ such that every symbol of a codeword can be recovered by accessing at most $r$ other symbols. The set of symbols participating in the recovery of a symbol is referred to as a {\it recovering set} (or {\it repair group}) of the symbol. Codes with small locality were introduced in~\cite{Huang:07, Han:07} (see also~\cite{Oggier:11}). The study of the locality property was inspired by the pioneering work of Gopalan \etal~\cite{Gopalan:12}. One of their key contributions was to establish a trade-off between the minimum Hamming distance of a  code and its locality, analogous to the classical Singleton bound. In particular, the authors showed that for a (scalar) linear $(n,k)$ code having locality $r$ for systematic symbols, its minimum distance $d$ is upper bounded as
\begin{equation}
\label{eq:Gopalan}
d \leq n - k - \left\lceil{\frac{k}{r}}\right\rceil + 2.
\end{equation}
They also demonstrated that the Pyramid code construction described in~\cite{Huang:07} achieves this bound. Since then, a series of papers have extended the distance bound for various types of codes, and have provided optimal code constructions that achieve the minimum distance bound (see, \eg,~\cite{Papailiopoulos:14,Rawat:14,Wang:15,Silberstein:13,TamoB:14,CadambeM:15,Sasidharan:15,Prakash:12}, and references therein). 

In this work, we focus our attention on a class of LRCs with multiple disjoint recovering sets~\cite{Wang:14,Rawat:14Availability,Tamo:14Availability,Pamies-Juarez:13}. Providing multiple disjoint recovering groups for symbols enables parallel reads and provides {\it high availability} of data. For this reason, codes with multiple disjoint recovering sets are referred to as codes with availability. Such codes are particularly attractive for data centers storing {\it hot data}, \ie, frequently accessed data. Moreover, they are useful in designing coded {\it private information retrieval}~\cite{Fazeli:15} and {\it locally rewriteable codes}~\cite{Kim:16}. 

A code is said to possess $\rt$-availability, if every symbol of a codeword has $t$ disjoint recovering sets each of size (\ie, locality) at most $r$. Most of the literature on codes with availability has been devoted to computing Singleton-like upper bounds on the minimum distance, and constructing codes with availability and large minimum distance. By comparison, relatively little has been said about bounds on the code rate, and constructions of high rate codes with availability. 
However, the authors of~\cite{Tamo:14Availability} (see also~\cite{Tamo:16bounds}) give an upper bound on the rate of codes with $\rt$-availability, and the authors of~\cite{Huang:16} give a field size dependent bound on the size of codes with availability, along the lines of~\cite{CadambeM:15}. Very recently, the authors of~\cite{Balaji:16bounds} presented an improved rate bound for $(r,t)$-availability, and in particular, for $(r,3)$-availability.

We are interested in computing tight upper bounds on the rate of LRCs with availability. Note that as we enhance the availability of the code by increasing the number of disjoint recovering sets, we are introducing more dependencies amongst the code symbols. Thus, the rate of a code with high availability cannot be too high, representing tension between high rate and high availability. We first focus our attention to binary LRCs (\ie, LRCs over $\GF{2}$) with availability $t = 2, 3$. We note that, in practice, small values of the availability parameter $t$ that are comparable to triple replication are the most interesting. Our motivation behind considering binary codes is that codes constructed over small finite fields, especially Galois fields of the form $\GF{2^m}$, are preferred in practice for their fast arithmetic~\cite{Plank:13}. 

Our main result is a uniqueness result for {\it rate optimal} codes for $t = 2, 3$. In essence, we show that for certain classes of binary linear codes with $(r,2)$ and $(2,3)$-availability, any {\it rate optimal} code must be a direct sum of shorter rate optimal codes. We note that designing a rate optimal code with availability can be viewed as a {\it covering problem}. In particular, when the $i$-th symbol of a code has $(r,t)$-availability, its dual code must contain $t$ codewords, each of weight at most $r+1$, such that their supports intersect only on $\{i\}$. We refer to such codewords as {\it covering codewords}. Designing a rate optimal code with $(r,t)$-availability is equivalent to finding a subspace of smallest dimension that contains covering codewords for all the symbols. It is worth noting that for covering problems, direct sum constructions are known to give good codes~\cite{Cohen:85}. %We wonder whether such a direct sum structure would be present in the rate optimal codes for other values of $r$ and $t$, and/or for codes over larger fields. 

Next, we obtain upper bounds on the rate of codes with $(2,t)$ and $(r,3)$-availability by using covering properties of their duals. In particular, we develop a novel method to bound the maximum weight of a coset leader of the dual code, which is known as the {\it covering radius} for linear codes (see~\cite{Cohen:85}), by using the its covering properties. This enables us to bound the number of cosets of the dual and get a rate upper bound. The method of bounding the number of cosets using covering properties may be of independent interest. 

Furthermore, we present a class of codes with $t = 2$ that are associated with convex polyhedra, in particular, the Platonic solids. We note that these codes associated with the Platonic solids may be of independent interest. We outline our contributions in the following section.

\subsection{Our Contributions}
\label{sec:contributions}
We highlight our broad contributions in the following. 
\begin{enumerate}
\item We first consider binary codes associated with convex polyhedra. More specifically, given a convex polyhedron $\Gamma$ with $e$ edges, fix an arbitrary labeling of its edges from $1$ to $e$. We define the code associated with a convex polyhedron\footnote{More generally, we can define a code associated with a planar graph in the same way.} $\Gamma$ as a subset $\code\subset\GF{2}^e$ such that for every vector $\cw \in \code$, the entries corresponding to edges that meet at a vertex of $\Gamma$ sum to zero over $\GF{2}$. In other words, vertices of $\Gamma$ define parity checks on the codewords of $\code$.  We demonstrate that such codes have $t = 2$, and that the codes associates with dual polyhedra are duals of each other. Further, we demonstrate that codes associated with the Platonic solids, namely, tetrahedron, octahedron, dodecahedron and icosahedron, are near-optimal in terms of their rates. %We study further properties of codes associated with the Platonic solids, namely, tetrahedron, octahedron, dodecahedron, and icosahedron. %We note that these codes associated with the Platonic solids may be of independent interest.

\item We focus on a class of binary $(n,k)$ codes $\class$ defined as the nullspace of an $N \times n$ parity-check matrix $H$, where each row has weight $r + 1$ and each column has weight $t$, such that $nt = N(r + 1)$. In addition, supports of any two rows of $H$ intersect in at most one point. We refer to codes in this class as codes with {\it exact covering}.\footnote{Codes in this class were also studied very recently in~\cite{Balaji:16bounds}, where such codes are referred to as codes with {\it strict availability}.}
\begin{enumerate}
\item First, we consider codes in $\class$ with $(r,2)$-availability. We show that, when $n\geq r+1$ and $r+1 \mid 2n$, the rate of $\code$ is upper bounded as $\frac{k}{n} \leq \frac{r}{r+2}$, with equality if and only if $\code$ is a direct sum of $\left[\frac{(r+1)(r+2)}{2},(r+1)\right]$ codes, each {\it generated} by the complete graph on $r+2$ points (Theorem~\ref{thm:t-2-exact-covering}).\footnote{The upper bound of $r/(r+2)$ has been shown in~\cite{Prakash:14}; see Remark~\ref{rem:sequential-recovery} for details.}

\item Next, we consider codes in $\class$ with $(2,3)$-availability. When the block length $n$ is a multiple of $7$, say $n = 7m$, we show that for any $\code\in\class$, we have $\rate{\code} \leq \frac{3}{7}$, with equality if and only if $\code$ is a direct sum of $m$ copies of the $\left[7,3\right]$ Simplex code (Theorem~\ref{thm:t-3-exact-covering}).
\end{enumerate}

\item We present novel rate upper bounds for codes in $\class$ with $(2,t)$ and $(r,3)$-availability (Theorem~\ref{thm:rate-bound-2-t} and Corollary~\ref{cor:rate-bound-r-3}). Our bounds for codes with $(2,t)$ and $(r,3)$-availability become sharper that the known bounds as the values of $t$ and $r$ increase, respectively.
\end{enumerate}

%Cite Balaji, other constructions of availability
%
%Local Recovery, availability (Dimakis, Tamo-Barg, Goparaju-Calderbank)
%Rate Optimality: when $t$ increases, rate decreases
%$t = 1, 2, 3$ Small $t$ is most interesting
%Philosophy: covering problems - are rate optimal codes direct products?
%Codes associated with platonic solids, examples, may be of independent interest
%Relation to local writeable codes, coded private information retrieval

\subsection{Relationship to Previous Work}
\label{sec:related-work}
{\bf Codes with availability:} The notion of multiple disjoint recovering sets has been studied in several works, see \eg~\cite{Wang:14,Rawat:14Availability,Tamo:14Availability,Pamies-Juarez:13, Huang:16,Wang:15ISIT,Balaji:16bounds}. 

{\it Rate Bounds:} The authors of~\cite{Tamo:14Availability} (see also~\cite{Tamo:16bounds}) show that for an $(n,k)$ code with $\rt$-availability, the rate is upper bounded as
\begin{equation}
\label{eq:Tamo-Barg-rate-bound}
\frac{k}{n} \leq \frac{1}{\prod_{j=1}^{t}\left(1 + \frac{1}{jr}\right)}.
\end{equation}
The authors of~\cite{Wang:14} and~\cite{Tamo:14Availability} also find upper bounds on and the minimum distance for codes with availability. Under suitable divisibility assumptions, these distance bounds can be  translated to the rate bounds. For $(r,2)$ availability the corresponding rate bounds from~\cite{Tamo:14Availability} and~\cite{Wang:14}, respectively, are
\begin{equation}
\label{eq:rate-bounds-derived}
\frac{k}{n} \leq \frac{r^{3} - r^3}{r^{3} - 1} + \frac{1}{n}, \quad \textrm{and} \quad \frac{k}{n} \leq \frac{2r-1}{2r+1} + \frac{1}{n(2r+1)}.
\end{equation}
For $(2,3)$-availability these derived rate bounds from~\cite{Tamo:14Availability} and~\cite{Wang:14} become $\left(\frac{8}{15}+\frac{1}{n}\right)$, and $\left(\frac{4}{7}+\frac{2}{7n}\right)$, respectively.  We note that our rate bounds for $(r,2)$ and $(2,3)$ availability strictly improve on these bounds. In a very recent work, the authors of~\cite{Balaji:16bounds} improve the bounds of~\cite{Tamo:14Availability}. We compare our bounds with the existing bounds in Section~\ref{sec:rate-bounds-comparison}.

%In~\cite{Wang:14}, the authors characterize an upper bound on the minimum distance of an $(n,k)$ code where only information symbols possess $\rt$-availability as
%\begin{equation}
%\label{eq:Wang-Zhang-distance-bound}
%d \leq n - k - \left\lceil\frac{t(k-1)+1}{t(r-1)+1}\right\rceil + 2.
%\end{equation}
%The above bound translates, under suitable divisibility assumptions, to the following rate bound.
%\begin{equation}
%\label{eq:Tamo-Barg-rate-bound-derived}
%\frac{k}{n} \leq \frac{t(r-1)+1}{tr+1} + \frac{t-1}{n(tr+1)}.
%\end{equation}

{\it Constructions:} It was noted in~\cite{Wang:14,Tamo:14Availability} that direct product codes possess availability property. The authors of~\cite{Kuijper:14} studied the availability property of simplex codes. The authors of~\cite{Goparaju:14} present a tensor-product based construction and a cyclic code construction for $r = 2$ and $t = 3$. The authors of~\cite{Huang:16} analyze the availability properties of a large number of well-known classical codes. Several code constructions using combinatorial structures are presented in~\cite{Rawat:14Availability, Pamies-Juarez:13, Wang:15ISIT}.
  
{\bf LRCs that locally correct multiple erasures:} We note that LRCs with $\rt$-availability form a class of codes that can correct any $t$ erasures locally. There are several classes of LRCs that can locally correct up to $t$ erasures as outlined below. Rate bounds on LRCs from any of these classes yield upper bounds on the rate of an LRCs with $\rt$-availability. A discussion on the hierarchy of these classes can be found in~\cite{Song:15}.

a) LRCs with {\it strong} local codes, wherein every symbol is protected by an $(r+t,r,t+1)$ local code, were considered in~\cite{Prakash:12,Silberstein:13,Rawat:14,Song:14}. For such codes, the rate is upper bounded as
\begin{equation}
\label{eq:strong-local-code-rate-bound}
\frac{k}{n} \leq \frac{r}{r+t}.
\end{equation}

b) LRCs with {\it cooperative} local recovery, wherein $t$ erasures can be simultaneously corrected by reading at most $r$ symbols, are considered in~\cite{Rawat:14Cooperative}. The rate bound of~\eqref{eq:strong-local-code-rate-bound} also applies to this family of codes.

c) LRCs with multiple repair alternatives are considered in~\cite{Pamies-Juarez:13}, wherein for any subset $E\subset[n]$ of size $t$, every symbol $i\in E$ can be recovered from at most $r$ symbols outside $E$. The authors present a family of such codes based on partial geometries, and give lower and upper bounds on the rate of codes in this family.

d) LRCs that allow sequential (or, successive) repair of $t$ erasures are considered in~\cite{Prakash:14,Song:15,Balaji:16,Balaji:16sequential}. The authors of~\cite{Prakash:14} present an upper bound on the rate of an $(n,k)$ code that allows sequential recovery of $t=2$ symbols with locality $r$ as
\begin{equation}
\label{eq:Prakash-rate-bound-t-2}
\frac{k}{n} \leq \frac{r}{r+2}.
\end{equation}
The authors also present optimal code construction based on Tur\'an graphs for a specific parameter range. They also demonstrate a code construction based on complete graphs, which has $(r,2)$-availability. %We essentially show the uniqueness of this construction for rate optimal codes with exact covering. 
Our result shows the uniqueness of such a construction for rate optimal codes with exact covering.
For the sequential recovery of $t=3$ erasures with locality $r$ under functional repair model,~\cite{Song:15} presents a lower bound on the length of a code. Under suitable divisibility assumptions, this bound for $r=2$ translates to the rate bound of $4/9$. 
%\begin{equation}
%\label{eq:Song-rate-bound-t-3}
%n \geq k + \left\lceil\frac{2k+\left\lceil\frac{k}{r}\right\rceil}{r}\right\rceil.
%\end{equation}
%This bound is further tightened in~\cite{Balaji:16} for $k \leq r^{1.8}$.
The authors of~\cite{Balaji:16} show a uniqueness result for rate optimal constructions for the sequential recovery from $t=2$ erasures.

{\bf Binary locally repairable codes:} A number of studies have recently considered LRCs over the binary field, see \eg,~\cite{Kuijper:14, Goparaju:14,Tamo:16,Wang:15ISIT, SilbersteinZ:15, ZehY:15,Huang:16}. In~\cite{Tamo:16} (see also~\cite{Goparaju:14}), in addition to presenting several code constructions, the authors also establish upper bounds on the rate of binary LRCs for various parameter regimes when $r = 2$. In addition, the authors present a direct sum of $[7,3]$ Simplex codes as an example of a larger code with $(2,3)$-availability. Our result shows rate optimality of such a construction and its uniqueness for the class of codes with exact covering.

{\bf Field size dependent bounds on the code dimension:} Field size dependent bounds on the minimum distance and rate for LRCs are considered in~\cite{CadambeM:15,AgarwalM:16a}. Simplex codes are shown to be rate optimal for $r=2$ amongst binary codes in~\cite{CadambeM:15}. The authors of~\cite{Huang:16} develop field size dependent bounds to incorporate the availability. 

\section{Preliminaries}
\label{sec:basics}

{\bf Notation:} We use the following notation. For an integer $l$, let $[l] = \{1,2,\ldots,l\}$. We use $\vect{x}(i)$ to denote the $i$-th coordinate of a vector $\vect{x}$, and $H(i,j)$ to denote the element in row $i$ and column $j$ in a matrix $H$. For a vector $\vect{x}$, $\supp{\vect{x}}$ denotes its support, \ie, $\supp{\vect{x}} = \{i : \vect{x}(i) \neq 0\}$. Let $\wt{\vect{x}}$ denote the Hamming weight of vector $\vect{x}$, \ie, $\wt{\vect{x}} = |\supp{\vect{x}}|$. For a set of vectors $\vect{x}_1, \ldots, \vect{x}_m$, $\subspace{\vect{x}_1, \ldots, \vect{x}_m}$ denotes their span; whereas for a matrix $H$, $\subspace{H}$ denotes its row space. For a vector space $\mathcal{A}$, $\dims{\mathcal{A}}$ denotes its dimension. For an $[n,k]$ code $\code$, its rate is denoted as $\rate{\code} = \frac{k}{n}$.

Let $\code$ denote a linear $\nkd$ code over $\GF{2}$ with block-length $n$, dimension $k$, and minimum distance $d$. Let $\cw$ denote a codeword in $\code$. 

We say that a code bit has availability $t$ with locality $r$ if it can be recovered from $t$ disjoint subsets of size at most $r$. The formal definition is as follows.

\begin{definition}
\label{def:locality}
[$\rt$-Availability] We say that the $i$-th bit of an $\nkd$ code $\code$ has $\rt$-availability if for any codeword $\cw\in\code$, there exist $t$ disjoint subsets $\rep{j}{i}\subset[n]\setminus \{i\}$, $|\rep{j}{i}|\leq r$, for $1\leq j\leq t$ such that $\cw(i) = \sum_{l\in\rep{j}{i}}\cw(l)$ for every $j\in[t]$. Each one of such subsets is referred to as a repair group for bit $i$. If every bit of $\code$ has $\rt$ availability, we say that $\code$ has $\rt$-availability. We  denote such a code as an $\nkdrt$ LRC.
\end{definition}
It is worth mentioning that repair groups of a bit can have different sizes, and we denote locality of the bit as the size of its largest repair group. Next, we focus our attention to codes associated with convex polyhedra.

\section{Codes associated with Convex Polyhedra}
\label{sec:Platonic-solids}

In this section, we present codes associated with convex polyhedra, and we focus on Platonic solids. Given a convex polyhedron, each edge corresponds to an entry of the codeword, and every vertex corresponds to a parity check. 

\begin{definition}
\label{def:code-generated-by-polyhedon}
Consider a convex polyhedron $\Gamma$ with $v$ vertices, $e$ edges, and $f$ faces. Fix an arbitrary labeling of its edges from $1$ through $e$. Let $\code$ be a subset of $\GF{2}^e$ such that for a vector $\cw\in\code$, the entries corresponding to edges that meet at a vertex sum to zero over $\GF{2}$. We say that the code $\code$ is generated by $\Gamma$, and denote it as $\code(\Gamma)$.
\end{definition}

\begin{definition}
\label{def:face-vector}
We say that a length-$N$ binary vector $\vect{v}$ corresponds to a face of $\Gamma$ if the locations of ones in $\vect{v}$ correspond to the edges forming that face. 
\end{definition}

First, we show that $\code(\Gamma)$ is a linear code generated by the faces of $\Gamma$.
%Consider a convex polyhedron $\Gamma$ with $M$ vertices, $N$ edges, and $K$ faces. %Let $\Gamma$ denote the graph associated with it. Note that $\Gamma$ must be a simple, 3-connected planar graph by Steinitz's theorem~\cite{}. 
%Fix an arbitrary labeling of edges from $1$ through $N$. Let $\code\subset\GF{2}^N$ such that for a vector $\cw\in\code$, the entries corresponding to edges that meet at a vertex sum to zero over $\GF{2}$.  
%Note that $\code$ is the kernel of the $M\times N$ incidence matrix $H$ of $\Gamma$, \ie, $\code = \{\cw\in\GF{2}^N \mid H\cw = 0\}$. In the following, we show that such vectors span the $\code$.

\begin{lemma}
\label{lem:convex-polyhedral-codes}
For a convex polyhedron $\Gamma$ with $v$ vertices, $e$ edges, and $f$ faces, the $\code(\Gamma)$ generated by $\Gamma$ is an $[e, f-1]$ linear code. Further, the vectors corresponding to the faces of $\Gamma$ span $\code(\Gamma)$. 
\end{lemma}
\begin{proof}
First, we prove that $\code$ is an $[e, f-1]$ linear code. Let us denote the graph formed by the edges and vertices of $\Gamma$ as $\Gamma'$. Note that $\code$ is the kernel of the $v\times e$ incidence matrix $H$ of $\Gamma'$, \ie, $\code = \{\cw\in\GF{2}^e \mid H\cw = 0\}$. Every column of $H$ has two ones, and thus the rows of $H$ sum to zero giving $\rnk{H} \leq v-1$. We show that there is no smaller linear dependency. 

Let $\vect{h}_{i}$ denote the row of $H$ corresponding to vertex $i$ of $\Gamma'$. Suppose, for contradiction, there is a smaller linear dependency $\sum_{i\in S} \vect{h}_{i} = 0$, where $S\subset[v]$. Now, every column of $H$ has exactly two ones corresponding to an edge of $\Gamma'$. Thus, for any vertex $i\in S$, all of its neighbors should be in $S$. However, as $\Gamma'$ is connected, $S$ must include all the $v$ vertices. Thus, $\rnk{H} = v-1$, and $\dims{\code} = e - v +1$. For a convex polyhedron, Euler's formula states that $v - e + f = 2$. Therefore, $\code$ is an $[e, f-1]$ linear code.  

Next, we show that $\code$ is generated by the vectors associated with faces. Let $G$ denote the matrix containing the vectors associated with the faces of $\Gamma$. Observe that each row of $G$ satisfies all the parity checks. Now, note that every column of $G$ corresponds to an edge of $\Gamma'$, and thus, has exactly two ones corresponding to the two faces that meet at that edge. Then, by applying the same arguments as in the case of $H$, we get that $\rnk{G} = f - 1$, and the result follows. 
\end{proof}

Recall that the two polyhedra are said to be (geometric) duals of each other if the vertices of one polyhedron correspond to the faces of the other, and vice-versa. Then, the following result follows from Lemma~\ref{lem:convex-polyhedral-codes}.

\begin{corollary}
\label{cor:dual-poly-dual-codes}
Let $\Gamma$ and $\Gamma^{\perp}$ be the dual convex polyhedra. Then, the dual code ${\dual(\Gamma)}$ of the code generated by a convex polyhedron $\Gamma$ is isomorphic to the code generated by its dual polyhedron $\Gamma^{\perp}$, \ie, \mbox{${\dual(\Gamma)} \cong \code(\Gamma^{\perp})$}.
\end{corollary}

\begin{lemma}
\label{lem:convex-polyhedron-availability}
The code generated by a convex polyhedron has $t = 2$ availability.
\end{lemma}
\begin{proof}
The value of the bit indexed by edge $\{u,v\}$ can be recovered by summing over $\GF{2}$ the entries of all the other edges incident either on vertex $u$, or vertex $v$. The edges incident on $u$ are disjoint from those incident on $v$.
\end{proof}

Next, we consider the codes associated with Platonic solids. Table~1 %\ref{tab:Codes-Platonic-solids} 
summarizes the codes associated with the Platonic solids. While specifying the parity check and generator matrices for these codes in the subsequent sections, we omit the zero entries for simplicity.

\begin{remark}
\label{rem:Platonic-codes}
As we show in the next section, the rate of a binary LRC with $(r,2)$-availability is upper bounded by $r/(r+2)$. Observe from Table~I that the code associated with tetrahedron is rate-optimal, whereas the rates of codes associated with cube, octahedron, dodecahedron, and icosahedron are near-optimal.
\end{remark}

\subsection{Tetrahedron Code}
\label{sec:tetra-code}
Fig.~\ref{fig:Platonic-solids}(a) shows the graph of the cube. Following the labeling of edges in Fig.~\ref{fig:Platonic-solids}(a), the set of parity checks can be written as
\begin{equation}
\label{eq:H-tetra}
H = 
\left[
\begin{array}{cccccc}
1 & 1 & {} & {} & 1 & {}\\
{} & 1 & 1 & {} & {} & 1\\
1 & {} & 1 & 1 & {} & {}\\
{} & {} & {} & 1 & 1 & 1
\end{array}
\right].
\end{equation}
Observe from~\eqref{eq:H-tetra} that the tetrahedron code has $(2,2)$-availability (see  Remark~\ref{rem:dual-code}). 
A generator matrix with rows corresponding to faces is given as
\begin{equation}
\label{eq:G-tetra}
G = 
\left[
\begin{array}{cccccc}
1 & {} & {} & 1 & 1 & {}\\
{} & 1 & {} & {} & 1 & 1\\
{} & {} & 1 & 1 & {} & 1\\
1 & 1 & 1 & {} & {} & {}
\end{array}
\right].
\end{equation}
It is easy to verify that $\rnk{H} = 3$, $\rnk{G} = 3$, and $GH^T = 0$. Therefore, the tetrahedron code is a $(6,3)$ code with $(2,2)$-availability. 

One can see that $G$ can be obtained from $H$ by first reordering the rows of $H$, row 1 $\rightarrow$ row 2 $\rightarrow$ row 3 $\rightarrow$ row 1, and then applying the permutation $(1 6)(2 4)(3 5)$ on the columns. Hence, the tetrahedral code is equivalent to its dual. 

\subsection{Cube Code and Octahedron Code}
\label{sec:cube-code}
Fig.~\ref{fig:Platonic-solids}(b) shows the graph of the cube. Following the labeling of edges in Fig.~\ref{fig:Platonic-solids}(b), the set of parity checks for the cube code can be written as
\begin{equation}
\label{eq:H-cube}
H = 
\left[
\begin{array}{cccccccccccc}
1 & {} & {} & 1 & 1 & {} & {} & {} & {} & {} & {} & {}\\
1 & 1 & {} & {} & {} & 1 & {} & {} & {} & {} & {} & {}\\
{} & 1 & 1 & {} & {} & {} & 1 & {} & {} & {} & {} & {}\\
{} & {} & 1 & 1 & {} & {} & {} & 1 & {} & {} & {} & {}\\
{} & {} & {} & {} & 1 & {} & {} & {} & 1 & {} & {} & 1\\
{} & {} & {} & {} & {} & 1 & {} & {} & 1 & 1 & {} & {}\\
{} & {} & {} & {} & {} & {} & 1 & {} & {} & 1 & 1 & {}\\
{} & {} & {} & {} & {} & {} & {} & 1 & {} & {} & 1 & 1
\end{array}
\right].
\end{equation}
From~\eqref{eq:H-cube}, observe that the cube code has $(2,2)$-availability (see Remark~\ref{rem:dual-code}).

%Notice that a codeword associated with a face of the cube satisfies all the parity checks. We give a matrix $G$ containing codewords corresponding to faces of the cube, each row corresponds to a face of the cube.
A generator matrix composed of vectors associated with the faces of the cube is as follows.
\begin{equation}
\label{eq:G-cube}
G = 
\left[
\begin{array}{cccccccccccc}
1 & 1 & 1 & 1 & {} & {} & {} & {} & {} & {} & {} & {}\\
1 & {} & {} & {} & 1 & 1 & {} & {} & 1 & {} & {} & {}\\
{} & 1 & {} & {} & {} & 1 & 1 & {} & {} & 1 & {} & {}\\
{} & {} & 1 & {} & {} & {} & 1 & 1 & {} & {} & 1 & {}\\
{} & {} & {} & 1 & 1 & {} & {} & 1 & {} & {} & {} & 1\\
{} & {} & {} & {} & {} & {} & {} & {} & 1 & 1 & 1 & 1
\end{array}
\right].
\end{equation}
%Note that the six rows sum to zero, and it is the only linear dependency amongst the rows of $G$. Hence, $G$ is a generator matrix of the cube code.
It is easy to verify that $\rnk{H} = 7$, $\rnk{G} = 5$, and $GH^T = 0$. Therefore, the code associated with the cube is a $(12, 5)$ code with $(2,2)$-availability. %, which also follows from Lemma~\ref{lem:convex-polyhedral-codes}.

\begin{figure}[!t]
\centering 
\subfigure[Tetrahedron]{
\includegraphics[scale=0.45]{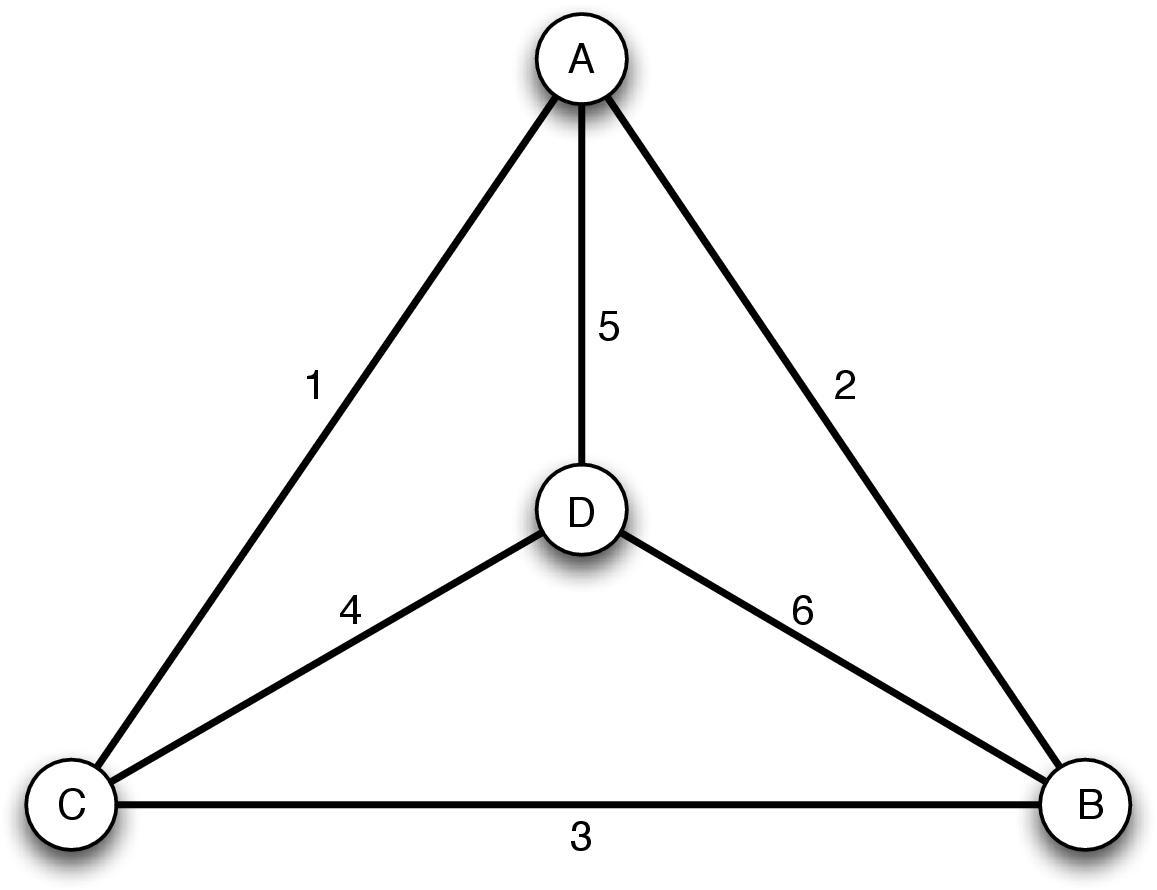}
}\\ 
\subfigure[Cube]{
\includegraphics[scale=0.45]{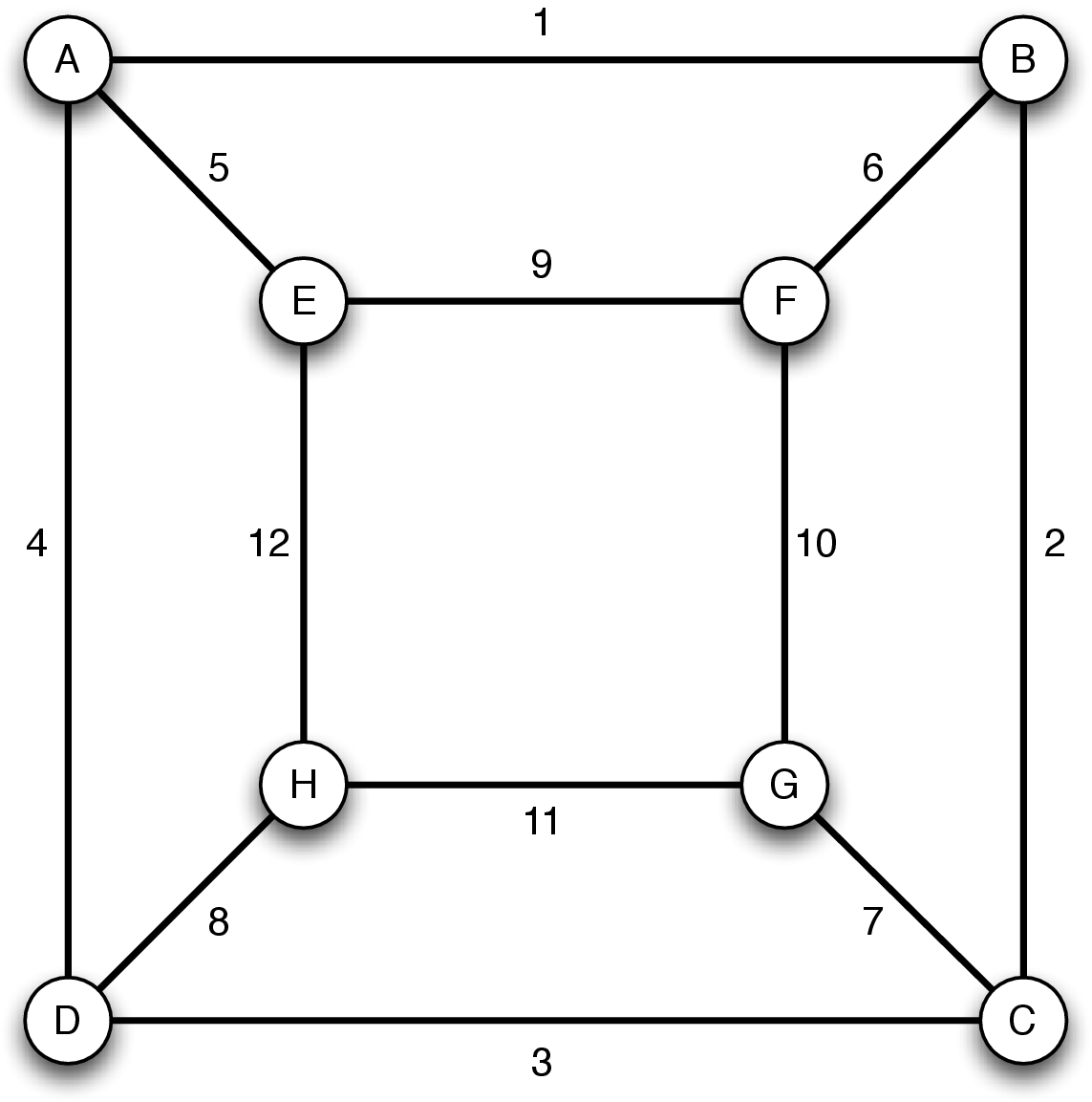}
}\qquad 
\subfigure[Octahedron]
{
\includegraphics[scale=0.45]{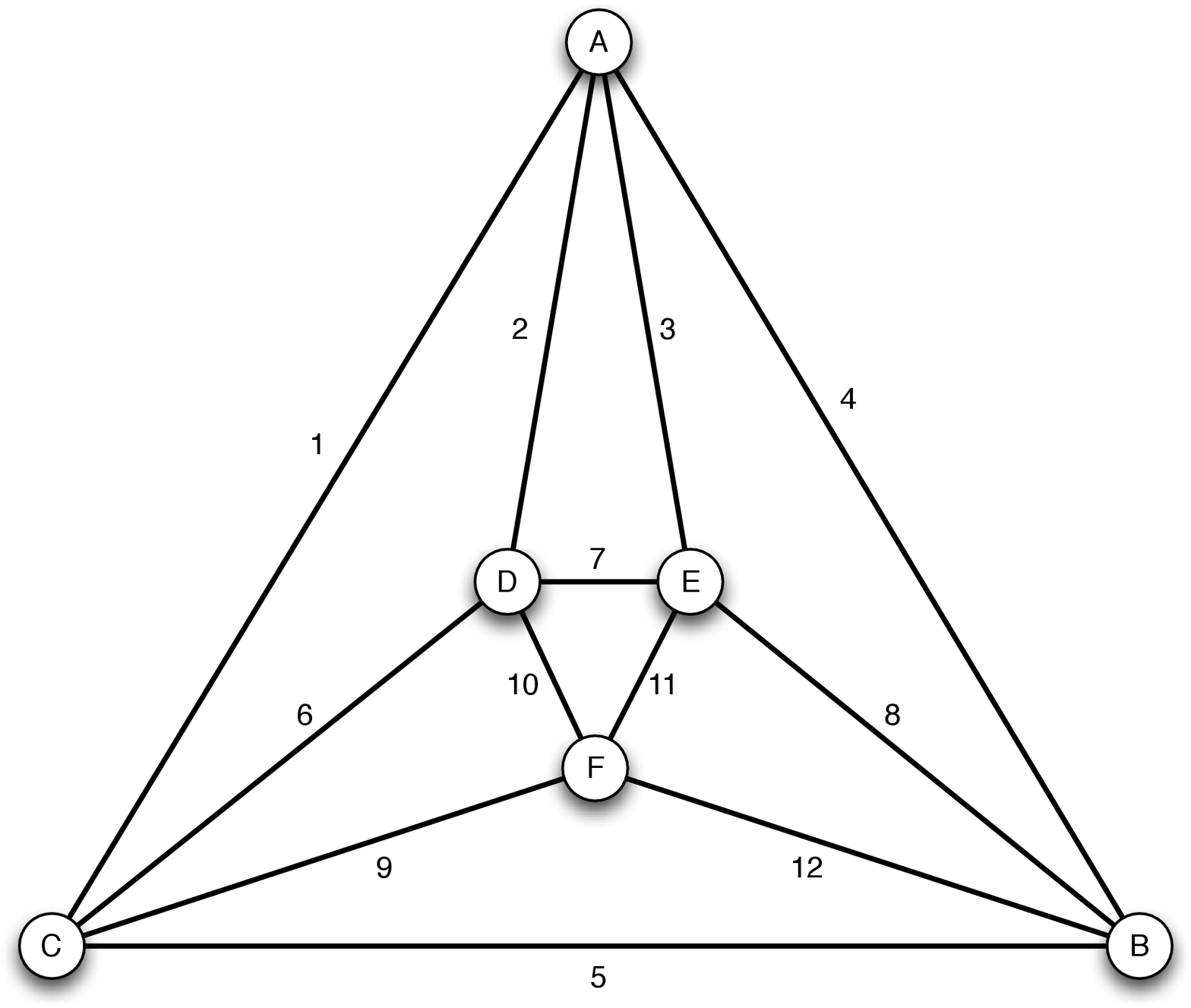}
}\\
\subfigure[Icosahedron]{
\includegraphics[scale=0.40]{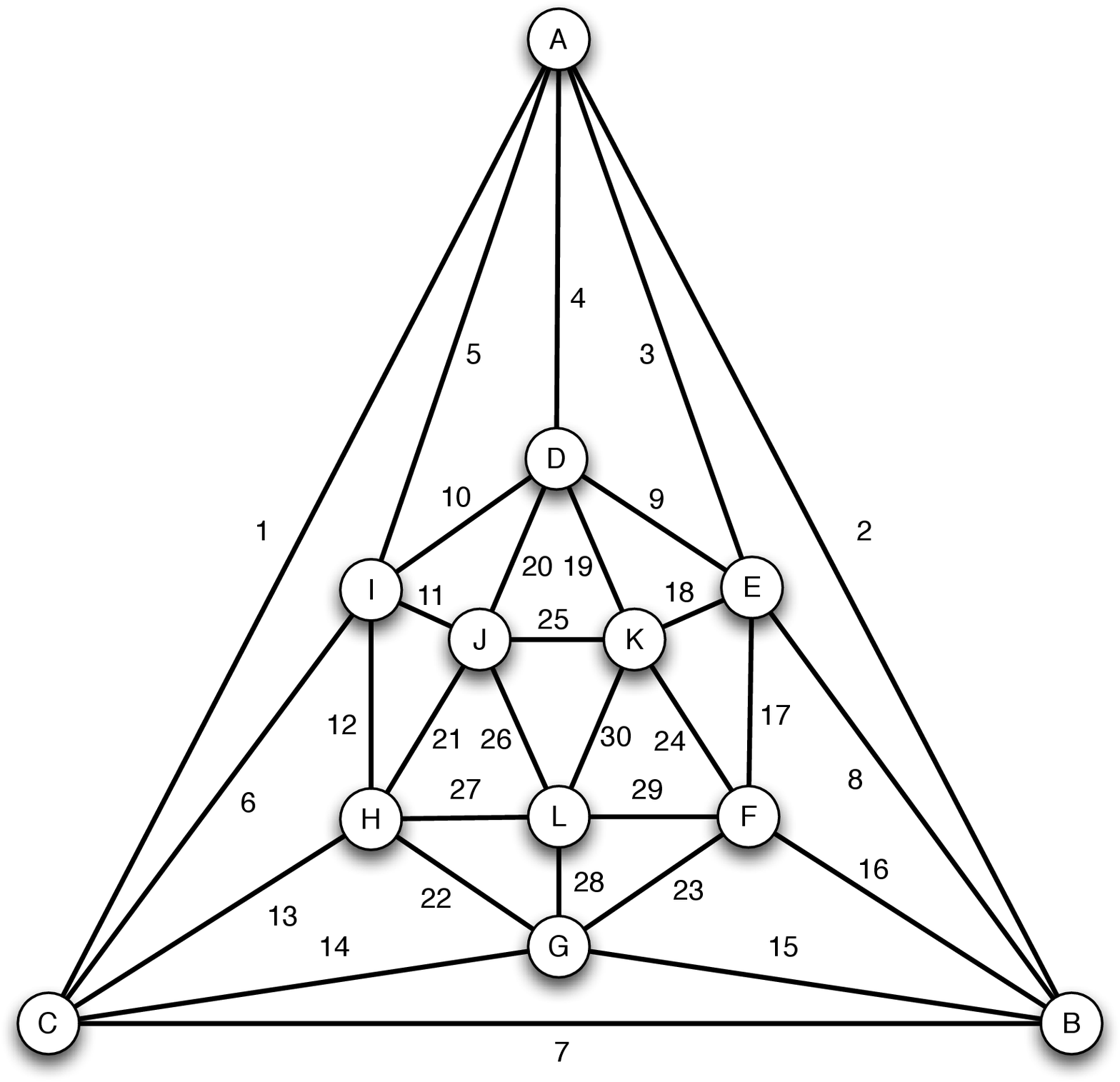}
}\qquad 
\subfigure[Dodecahedron]
{
\includegraphics[scale=0.40]{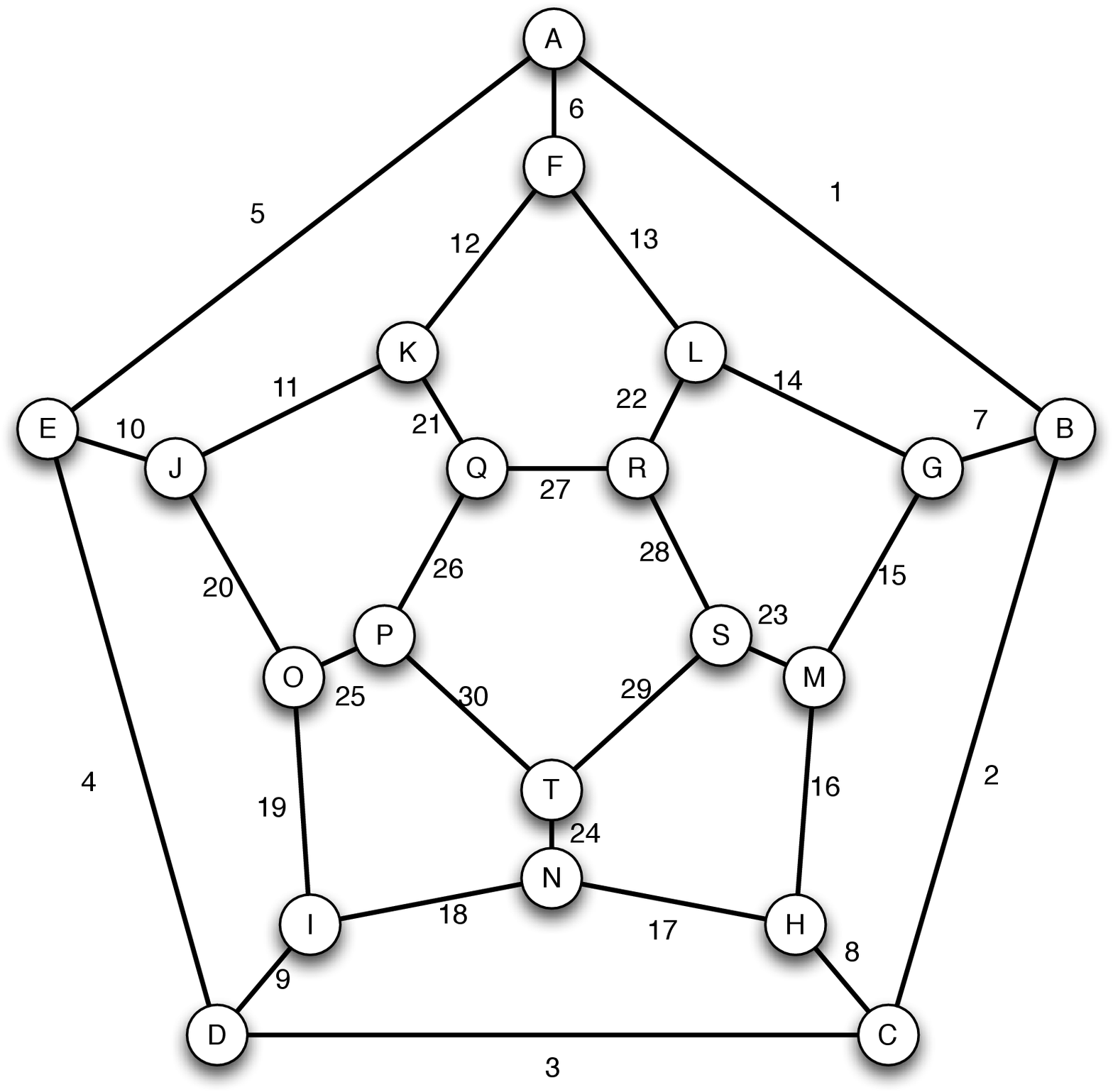}
}
\caption{Graphs associated with the Platonic solids.} 
\label{fig:Platonic-solids}
\end{figure}

%\subsection{Octahedron Code}
%\label{sec:octa-code}
%The octahedron code is associated with graph of the octahedron (see Fig.~\ref{fig:Platonic-solids}(c)). 
Recall that the cube and the octahedron are geometric duals of each other. A parity check matrix $H$ of the octahedron code is given below. We follow the labeling of edges in as shown in Fig.~\ref{fig:Platonic-solids}(c).
\begin{equation}
\label{eq:H-octa}
H = 
\left[
\begin{array}{cccccccccccc}
1 & 1 & 1 & 1 & {} & {} & {} & {} & {} & {} & {} & {}\\
{} & {} & {} & 1 & 1 & {} & {} & 1 & {} & {} & {} & 1\\
1 & {} & {} & {} & 1 & 1 & {} & {} & 1 & {} & {} & {}\\
{} & 1 & {} & {} & {} & 1 & 1 & {} & {} & 1 & {} & {}\\
{} & {} & 1 & {} & {} & {} & 1 & 1 & {} & {} & 1 & {}\\
{} & {} & {} & {} & {} & {} & {} & {} & 1 & 1 & 1 & 1
\end{array}
\right].
\end{equation}
From~\eqref{eq:H-octa}, observe that the octahedron code has $(3,2)$-availability (see Remark~\ref{rem:dual-code}).
%Note that the six rows sum to zero, and thus, are linearly dependent. It is easy to check that this is the only linear dependency amongst the rows of $H$. Hence, $\rnk{H} = 5$, and the code is a $[12,7]$ code. 

From~\eqref{eq:G-cube} and~\eqref{eq:H-octa}, we see that the cube code and the octahedron code are duals of each other. 

A generator matrix of the octahedron code can be given by~\eqref{eq:H-cube}. Note that each row of $H$ in~\eqref{eq:H-cube} corresponds to a face of the octahedron.

\subsection{Icosahedron Code and Dodecahedron Code}
\label{sec:icosa-code}
%The icosahedron code is associated with graph of the octahedron (see Fig.~\ref{fig:Platonic-solids}(c)).
One can easily find a parity check matrix of the icosahedron code following the edge labeling in Fig.~\ref{fig:Platonic-solids}(d). Observe that the icosahedron code has $(4,2)$-availability. A generator matrix with its rows as the faces of the icosahedron can be easily computed from Fig.~\ref{fig:Platonic-solids}(e). One can check that $\rnk{G} = 11$, and the icosahedron code is a $(30,11)$ code. Recall that the icosahedron and the dodecahedron are geometric duals of each other. %The code associated with the dodecahedron has $H$ in~\eqref{eq:H-icosa} as a generator matrix and $G$ in~\eqref{eq:G-icosa} as a parity check matrix. 
The dodecahedron code is a $(30,19)$ code with $(2,2)$-availability.

\begin{table}[!t]
\begin{center}
\begin{tabular}{|c|c|c|}
\hline
{\bf Polyhedron and its Dual} & {\bf Associated Code} & {\bf Weight Enumerator}\\
\hline
\hline
Tetrahedron & \begin{tabular}{@{}c@{}} $[6,3]$ code, \\ $(2,2)$-availability \end{tabular} & $1+4z^3+3z^4$\\
\hline
Tetrahedron & \begin{tabular}{@{}c@{}} $[6,3]$ code, \\ $(2,2)$-availability \end{tabular} & $1+4z^3+3z^4$\\
\hline
\hline
Cube &\begin{tabular}{@{}c@{}} $[12,5]$ code, \\ $(2,2)$-availability \end{tabular} & $1+6z^4+16z^6+9z^8$\\
\hline
Octahedron & \begin{tabular}{@{}c@{}} $[12,7]$ code, \\ $(3,2)$-availability \end{tabular} & $1+8z^3+15z^4+24z^5+32z^6+24z^7+15z^8+8z^9+Z^{12}$\\
\hline
\hline
Dodecahedron & \begin{tabular}{@{}c@{}} $[30,11]$ code, \\ $(2,2)$-availability \end{tabular} & \begin{tabular}{@{}c@{}} $1+20z^3+30z^4+72z^{5}+400z^{6}+1140^{7}+2715z^{8}$\\ $+6560z^{9}+14112z^{10}+26280z^{11}+42740z^{12}+59760z^{13}$ \\  $+72000z^{14}+75912z^{15}+70215z^{16}+57120z^{17}+41440z^{18}$ \\ $+26820z^{19}+15246z^{20}+7560z^{21}+3120z^{22}+900z^{23}+125z^{24}$ \end{tabular}\\
\hline
Icosahedron & \begin{tabular}{@{}c@{}} $[30,19]$ code, \\ $(4,2)$-availability \end{tabular} & \begin{tabular}{@{}c@{}} $1+12z^5+30z^8+20z^{9}+72z^{10}+120z^{11}+100z^{12}+180z^{13}$ \\ $+240z^{14}+272z^{15}+345z^{16}+300z^{17}+200z^{18}+120z^{19}+36z^{20}$ \end{tabular}\\

\hline
\end{tabular}
\caption{Codes associated with the Platonic solids}
\end{center}
\label{tab:Codes-Platonic-solids}
\end{table}

\section{Rate-Optimal Codes with Small Availability}
\label{sec:rate-optimal}
We are interested in rate-optimal codes with $(r,t)$-availability, which are defined as follows. 

\begin{definition}
\label{def:rate-opti}
[Rate Optimality] A code $\code$ with $\rt$-availability is said to be rate optimal if its rate is maximum among all (binary, linear) codes possessing $\rt$-availability.
\end{definition}

It is straightforward to see that the $\rt$-availability of a code $\code$ imposes certain constraints on its dual code $\dual$ in the following way.
\begin{remark}
\label{rem:dual-code}
The $i$-th bit of a code $\code$ has $\rt$ availability if and only if its dual code $\dual$ contains $t$ codewords $\repcw{i}{1}, \repcw{i}{2}, \ldots, \repcw{i}{t}$ such that for all $l\in[t]$, $i\in\supp{\repcw{i}{l}(i)}$, $|\supp{\repcw{i}{l}}| \leq r+1$, and for all $p,q\in[t], p\neq q$, $\supp{\repcw{i}{p}}\cap\supp{\repcw{i}{q}} = \{i\}$. We call such $t$ codewords as repair codewords for the $i$-th bit.
\end{remark}

In other words, the availability requirement of a code places constraints on the supports of certain codewords in the dual code. Our central idea is to carefully analyze the structure of the dual code to obtain upper bounds on the rate of the code with availability. 

For simplicity of notation, we refer to the coordinates as {\it points}, and represent every codeword by its support. In particular, we refer to a weight $w$ codeword as a $w$-{\it subset} of $[n]$ (or just as a subset if its Hamming weight is clear from the context or if it is not important). For analyzing the structure of the dual code, we use notions of {\it covering} and {\it covering with $(r,t)$-availability}, defined as follows.

\begin{definition}
\label{def:covering}
[Covering] We say that a $w$-subset $S$ covers point $i$, if $i\in S$. Further, we say that a code $\code$ covers point $i$ $l$ times, if $\code$ contains $l$ subsets that cover point $i$. 
\end{definition}

%Next, we introduce the notion of {\it covering with $(r,t)$-availability}.
 
\begin{definition}
\label{def:t-covering}
[Covering with $(r,t)$-Availability] %We say that a code $\code$ covers point $i$ with $(r,t)$-availability, if $\code$ covers point $i$ (at least) $t$ times with subsets of size at most $r+1$ such that the the subsets covering $i$ intersect only on $i$. We call the subsets that cover a point with availability $t$ as its $t$-covering subsets. 
We say that a code $\code$ covers point $i$ with $(r,t)$-availability, if $\code$ covers point $i$ (at least) $t$ times such that the subsets covering $i$ are of size at most $r+1$ and they intersect only on $i$. We call such subsets as $t$-{\it covering subsets} (or, simply, as {\it covering subsets}).
\end{definition} 

We can restate Remark~\ref{rem:dual-code} in terms of covering as follows.
\begin{remark}
\label{rem:covering}
A code $\code$ has $(r,t)$-availability if and only if its dual code $\dual$ covers each of the $n$ points with $(r,t)$-availability.
\end{remark}

Finally, we introduce the notion of the code generated by a graph $\Gamma$ (similar to the code generated by a convex polyhedron). 
\begin{definition}
\label{def:code-generated-by-graph}
Consider a planar graph $\Gamma$ with $v$ vertices and $e$ edges. Fix an arbitrary labeling of its edges from $1$ through $e$. Let $\code$ be a subset of $\GF{2}^e$ such that for every vector $\cw\in\code$, the entries of $\cw$ corresponding to edges that meet at a vertex sum to zero over $\GF{2}$. We say that the code $\code$ is generated by $\Gamma$, and denote it as $\code(\Gamma)$.
\end{definition}
Notice that $\code$ is a linear code with the incidence matrix of $\Gamma$ as its parity check matrix.

In the remaining of the paper, we denote the code containing the covering subsets as the primal code $\code$. Note that its dual code $\dual$ possesses the $\rt$-availability property.

\section{Codes with $(r,2)$-Availability}
\label{sec:2-availability}
Our focus, in this section, is on the codes in which each bit can be recovered from two disjoint recovering sets each of size at most $r+1$. From Remark~\ref{rem:covering}, notice that the primal code $\code$ should cover every point with $(r,2)$-availability. From simple counting arguments, it follows that to cover $n$ points with $(r,2)$-availability, $\code$ should contain at least $\frac{2n}{r+1}$ subsets of size up to $r+1$. First, we consider the case when $\code$ contains exactly $\frac{2n}{r+1}$ $2$-covering subsets, each of size $r+1$.

\subsection{Exact Number of Covering Subsets of Size $r+1$}
\label{sec:t-2-exact-covering}

\begin{theorem}
\label{thm:t-2-exact-covering}
Let $n$ and $r$ be non-negative integers such that $n\geq r+1$ and $r+1\mid 2n$. Let $\code$ be the length-$n$ primal code spanned by $\frac{2n}{r+1}$ $(r+1)$-subsets that cover every point with $(r,2)$-availability. Then, the rate of its dual code $\dual$ is upper bounded as $\rate{\dual} \leq \frac{r}{r+2}$, with equality if and only if $\dual$ is (equivalent to) a direct sum of $\left[\frac{(r+1)(r+2)}{2},(r+1)\right]$ codes, each of which is the code generated by the complete graph on $r+2$ points.
\end{theorem}
\begin{proof}
Let $\set{S}$ be the set of $N = \frac{2n}{r+1}$ covering $(r+1)$-subsets. Label them (in arbitrary order) as $S_1, \cdots, S_N$.  Form a graph $\Gamma$ with $N$ vertices, where every vertex corresponds to a covering $(r+1)$-subset. Join vertices $i$ and $j$ if the corresponding $(r+1)$-subsets $S_i$ and $S_j$ intersect. Observe that a pair of covering subsets can intersect in at most one point as there are exactly $\frac{2n}{r+1}$ of them. Moreover, since $\set{S}$ covers every point exactly twice, each vertex in $\Gamma$ has degree $r+1$.  

If for some $\set{T}\subseteq\set{S}$, $\sum_{j\in\set{T}}S_j = 0$, then the vertices of $\Gamma$ corresponding to subsets in $\set{T}$ determine a connected component of $\Gamma$. %as every point is covered exactly twice. 
Note that the size of a connected component in $\Gamma$ is at least $r+2$ as $\Gamma$ is an $(r+1)$-regular graph.

Now, partition $\Gamma$ into connected components, and eliminate a vertex from every connected component. This yields at least $\frac{r+1}{r+2}N = \frac{2n}{r+2}$ vertices such that $(r+1)$-subsets corresponding to these vertices are linearly independent. Therefore, $\dims{\code} \geq \frac{2n}{r+2}$, and the upper bound on $\rate{\dual}$ follows. In addition, we have $\dims{\code} = \frac{2n}{r+2}$  if and only if the connected components of $\Gamma$ are complete graphs of of size $r+2$. This essentially specifies that the incidence matrix of $\Gamma$ has a block diagonal structure with each block being the incidence matrix of the complete graph on $r+2$ points. Hence, a rate optimal $\dual$ must be a direct sum of the codes generated by complete graphs on $r+2$ points.
\end{proof}

\begin{remark}
\label{rem:sequential-recovery}
The upper bound of $r/(r+2)$ on the rate of any linear code with $(r,2)$-availability has been  established in~\cite{Prakash:14} by considering a broader class of codes that allow {\it sequential} recovery of $2$ symbols with locality $r$. Further, the authors note that the code associated with the complete graph on $r+2$ vertices is a rate-optimal code with $(r,2)$-availability. Clearly, a direct sum of codes associated with complete graph on $r+2$ vertices is also rate-optimal. Theorem~\ref{thm:t-2-exact-covering} shows the uniqueness of such a construction for achieving rate-optimality in binary codes with $(r,2)$-availability.
\end{remark}

\subsection{Exact Number of Covering Subsets of Multiple Sizes}
\label{sec:t-2-exact-multiple-sizes}

\begin{corollary}
\label{cor:t-2-exact-covering-multiple-size}
Let $n$ and $r$ be non-negative integers such that $n\geq (r+1)$. Let $\code$ be the length-$n$ primal code spanned by $N$ subsets of multiple sizes wth maximum size $r+1$, which cover every point exactly twice with availability. Then, the rate of its dual code $\dual$ is upper bounded as $\rate{\dual} < \frac{r}{r+2}$.
\end{corollary}
\begin{proof}
Let $\Nj{j}$ be the number of $2$-covering $j$-subsets for $1\leq j\leq r+1$. Since each of the $n$ points is covered exactly twice, we have 
\begin{equation}
\label{eq:n-multiple-sizes}
n = \frac{\sum_{j=1}^{r+1}j\Nj{j}}{2}.
\end{equation}
 
The proof essentially follows the same argument as the proof of Theorem~\ref{thm:t-2-exact-covering}. Form a graph $\Gamma$ with $N$ subsets as vertices, wherein a pair of vertices are adjacent if the corresponding subsets intersect. % let $\set{S}\subseteq\code$ be the set of covering subsets. Label them as $S_1, \cdots, S_N$.  Form a graph $\Gamma$ with $N$ subsets as vertices, put an edge between vertices $l$ and $m$ if the corresponding subsets $S_l$ and $S_m$ intersect. %since $\set{S}$ covers every point exactly twice, a pair of covering subsets can intersect in at most one point. Further, the degree of a vertex in $\Gamma$ equals the size of the corresponding subset.  

Now, a minimal linear dependency amongst the covering subsets determines a connected component of $\Gamma$, as every point is covered exactly twice. 
Partitioning $\Gamma$ into connected components, and eliminating a vertex from every connected component, we get a lower bound on the dimension of $\code$ as $\dims{\code} \geq \sum_{j=1}^{r+1}\frac{j}{j+1}\Nj{j}$. This follows since the size of a connected component of $\Gamma$ containing a vertex corresponding to a $j$-subset is at least $j+1$.
Clearly, $\sum_{j=1}^{r+1}\frac{j}{j+1}\Nj{j} \geq \frac{1}{r+2}\sum_{j=1}^{r+1}j\Nj{j}$ with strict inequality when there is a covering subset of size less than $r+1$. Then, from~\eqref{eq:n-multiple-sizes}, we have $\dims{\code} > \frac{2n}{r+2}$ from which the result follows.
\end{proof}

\section{Codes with $(2,3)$-Availability}
\label{sec:3-availability}

In this section, we focus on the codes with $r = 2$ and $t = 3$. Simple counting arguments show that to cover $n$ points with $(2,3)$-availability, the primal code $\code$ must contain at least $n$ $3$-covering subsets of size up to $3$. We consider the case of of {\it exact covering}, wherein $\code$ contains exactly $n$ $3$-covering 3-subsets. For the case when the block-length is a multiple of $7$, we show that the code rate is upper bounded by $3/7$, and prove that any rate optimal code needs to be a {\it direct sum (or tensor-product)} style construction. The statement of the result is as follows.

%\subsection{Exact Number of Covering Subsets of Size $3$}
%\label{sec:t-2-exact-covering}

\begin{theorem}
\label{thm:t-3-exact-covering}
For a positive integer $m$, let $n = 7m$. Let $\code$ be the length-$n$ primal code spanned by  $7m$ $3$-subsets that cover every point with $(2,3)$-availability. Then, we have $\rate{\dual} \leq \frac{3}{7}$, with equality if and only if $\dual$ is (equivalent to) a direct sum of $m$ copies of the $\left[7,3\right]$ Simplex code.
\end{theorem}

\begin{remark}
\label{rem:Simplex-codes}
Simplex codes have been shown to be rate optimal for $r = 2$ amongst binary codes in~\cite{CadambeM:15}. Several constructions based on Simplex codes have been proposed, \eg,~\cite{Kuijper:14,Goparaju:14,Tamo:16,ZehY:15}. The authors of~\cite{Goparaju:14} present a direct sum of [7, 3] Simplex codes as an example of a code with $(2,3)$-availability. Theorem~\ref{thm:t-3-exact-covering} shows the uniqueness of such a construction for achieving rate optimality in binary codes with $(2,3)$-availability.
\end{remark}

\subsection{Proof of Theorem~\ref{thm:t-3-exact-covering}}
\label{sec:t-3-exact-covering-proof}
%\begin{proof}
The steps involved in the proof are outlined below.
\begin{enumerate}
\item First, we show that $\code$ must contain at least $m$ pairwise disjoint covering $3$-subsets. 
\item Next, we prove that $\dims{\dual} \leq 3m$, and the equality occurs if and only if the size of a maximum set of pairwise disjoint covering $3$-subsets in $\code$ is exactly $m$. 
%	\begin{enumerate}
%		\item 
To prove this, we first assume that there exists a maximum set of pairwise disjoint covering 3-subsets in $\code$ of size $\mip$ for some non-negative integer $i'$. Then, we show that $\dims{\dual}$ is strictly less than $3m$ if $i' > 0$.
%	\end{itemize}
\item Finally, we prove that, if $\dims{\code} = 4m$, then the size of a maximum collection of pairwise disjoint covering $3$-subsets in $\code$ is exactly $m$, and $\code$ must be (equivalent to) a direct sum of $m$ copies of a $[7,4]$ Hamming code. 
\end{enumerate}

\subsubsection{Step 1}
\label{sec:step-1}
\begin{lemma}
\label{lem:disjoint-3-subsets}
For a positive integer $m$, let $n = 7m$. Let $\code$ be the length-$n$ primal code spanned by  $7m$ $3$-subsets that cover every point exactly thrice with availability. Then, $\code$ must contain at least $m$ pairwise disjoint 3-subsets. %If $\code$ contains exactly $m$ pairwise disjoint 3-subsets, then $\code$ must be a direct sum of the $m$ copies of the $[7,4]$ Hamming code.
\end{lemma}
\begin{proof}
Label the $n$ covering 3-subsets as $S_1, \cdots, S_n$. Form a graph $\Gamma$ with $n$ vertices, where every vertex corresponds to a covering $3$-subset. Put an edge between vertices $i$ and $j$ if the corresponding $3$-subsets $S_i$ and $S_j$ intersect. Since every point is covered exactly thrice, $\Gamma$ must be a 6-regular graph. 

Now, a set of pairwise disjoint covering 3-subsets determine an independent set in $\Gamma$. For a $j$-regular graph of order $n$, the size of an independent set is at least $\left\lceil\frac{n}{j+1}\right\rceil$ (see~\cite[Theorem 1]{Rosenfeld:64}), from which the result follows.
\end{proof}

\begin{figure}[!t]
\centering 
\includegraphics[scale=0.75]{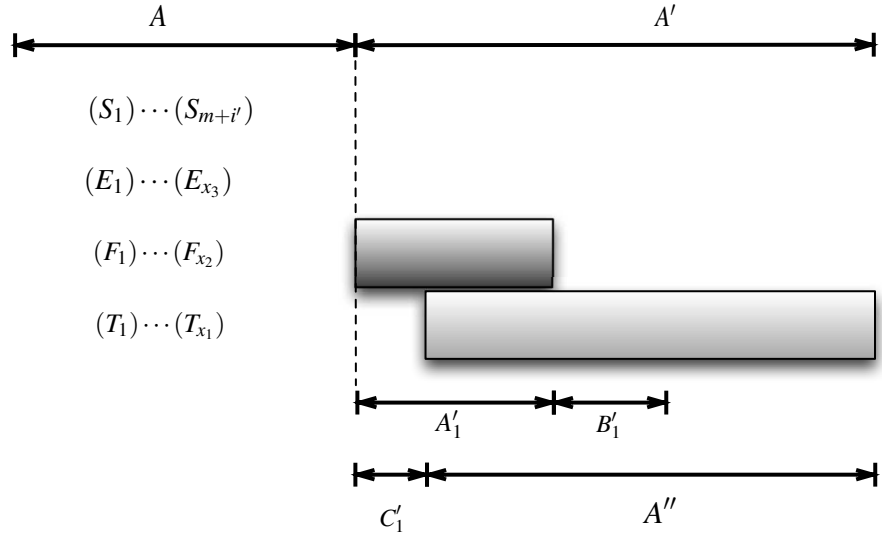}
\caption{Schematic depicting the notation for the step 2 in the proof of Theorem~\ref{thm:t-3-exact-covering}.}
\label{fig:schematic-3-subsets}
\end{figure}

\subsubsection{Step 2}
\label{sec:step-2}
We begin with establishing the key ingredients that aid in this step. %We present a schematic depicting these ingredients in Fig.~\ref{fig:schematic-3-subsets} for the ease of understanding.
\begin{enumerate}
\item {\bf A maximum set of pairwise disjoint 3-subsets in $\code$:} Suppose the size of a maximum set of pairwise disjoint covering 3-subsets in $\code$ is $\mip$. We label these subsets as $S_1,\ldots,S_{\mip}$. Let $\A$ be the set of points covered by these subsets, \ie, $\A = \cup_{j=1}^{\mip}S_j$. Let $\Ap = [n] \setminus \A$. See Fig.~\ref{fig:schematic-3-subsets} for the ease of understanding. Note that $|\A| = 3m + 3i'$ and $|\Ap| = 4m - 3i'$.

\item {\bf Three {\emph{types}} of 3-subsets depending on their intersection with $\A$:} Let $\xsubi{i}$ be the number of 3-subsets that intersect $A$ in $i$ points for $1\leq i\leq 3$. (By maximality of $S_1,\ldots,S_{\mip}$, $\xsubi{0} = 0$.) Let $\{\Esubj{j} : 1\leq j\leq \xsubi{3}\}$, $\{\Fsubj{j} : 1\leq j\leq \xsubi{2}\}$, and $\{\Tsubj{j} : 1\leq j\leq \xsubi{1}\}$ be the collections of 3-subsets that meet $\A$ in 3, 2, and 1 points, respectively. 

Let $\App$ be the set of points in $\Ap$ that are covered by the type $T$ subsets, \ie, $\App = \Ap \cap \left(\cup_{j=1}^{\xsubi{1}}\Tsubj{j}\right)$. Let $\Aonep$ be the points in $\Ap$ that are covered by the type $F$ subsets, \ie, $\Aonep = \Ap \cap \left(\cup_{j=1}^{\xsubi{2}}\Fsubj{j}\right)$. Let $\Conep\subseteq\Aonep$ be the set of points that are covered only by the type $F$ subsets. These sets are depicted schematically in Fig.~\ref{fig:schematic-3-subsets}.

\item {\bf \emph{Singletons} and \emph{pairs} of points:} Consider the multiset of points in $\Aonep$ that are covered by the type $F$ subsets. We refer to the elements of this multiset as {\it singletons}. Note that the size of this multiset is $\xsubi{2}$. Similarly, consider the multiset of points in $\App$ that are covered by the type $T$ subset. Every type $T$ covers two points from $\App$, which are referred to as a {\it pair} (of points). There are $\xsubi{1}$ such pairs in the multiset.

\item {\bf Graph $\Gamma$ formed on pairs:} Form a graph $\Gamma$ by  assigning a vertex corresponding to every point in $\App$, and adding an edge between two vertices if they correspond to a pair. % of points that belong to the same type $T$ 3-subset. 
Note that the number of vertices of $\Gamma$ is $|\App|$, and the number of edges in $\Gamma$ is $\xsubi{1}$.

Partition $\Gamma$ into connected components. Let $\Bonep$ be the vertices of the connected components of $\Gamma$ that (directly or indirectly) touch $\Aonep$. In other words, $\Bonep \subseteq \Ap \setminus \Aonep$ be the set of points such that any vertex corresponding to a point in $\Bonep$ is connected to a vertex corresponding to a point in $\Aonep$. Again, refer to Fig.~\ref{fig:schematic-3-subsets} for a schematic representation of $\Bonep$.

\item {\bf Analysis of the singletons and pairs:} Suppose that a fraction $f\xsubi{2}$ of singletons touch connected components of $\Gamma$. Note that these $f\xsubi{2}$ are the singletons in $\Aonep\setminus\Conep$, and we have 
\begin{equation}
\label{eq:C1-prime-size}
|\Conep| = \frac{(1-f)\xsubi{2}}{3}.
\end{equation} 

Next, suppose that a fraction $g|\Aonep\setminus\Conep|$ of points in $\Aonep\setminus\Conep$ have degree one in $\Gamma$ and the remaining $(1-g)|\Aonep\setminus\Conep|$ points have degree two in $\Gamma$. Consider the multiset of points with indices in $\Aonep\setminus\Conep$ that are covered by the type $F$ and type $T$ subsets. There are $3|\Aonep\setminus\Conep|$ such points, of which, $f\xsubi{2}$ are covered by the type $F$ subsets, $g|\Aonep\setminus\Conep|$ are covered once by the type $T$ subsets, and $(1-g)|\Aonep\setminus\Conep|$ are covered twice by the type $T$ subsets. Therefore, 
$$3|\Aonep\setminus\Conep| = f\xsubi{2} + g|\Aonep\setminus\Conep| + 2(1-g)|\Aonep\setminus\Conep|,$$
which yields
\begin{equation}
\label{eq:A1-prime-minus-C1-prime-size}
|\Aonep\setminus\Conep| = \frac{f}{1+g}\xsubi{2}
\end{equation}

\end{enumerate}

For the simplicity of notation, denote the dual code $\dual$ as $\codeD$. To obtain an upper bound on the dimension of $\codeD$, consider the projection of $\codeD$ on $\A \cup \Aonep \cup \Bonep$, denoted as $\codeproj{\codeD}{\A \cup \Aonep \cup \Bonep}$, and its kernel, denoted as $\codeDp$, which is the subcode of $\codeD$ that vanishes on $\A \cup \Aonep \cup \Bonep$. Now, by the rank-nullity theorem, we have 
\begin{equation}
\label{eq:rank-nullity-D}
\dims{\codeD} = \dims{\codeDp} + \dims{\codeproj{\codeD}{\A \cup \Aonep \cup \Bonep}}.
\end{equation}
In the following, we obtain an upper bound on the dimensions of $\codeDp$ and $\codeproj{\codeD}{\A \cup \Aonep \cup \Bonep}$.

\begin{lemma}
\label{lem:dimension-D-prime}
Let $\codeDp$ be the subcode of $\codeD$ that vanishes on $\A \cup \Aonep \cup \Bonep$. Then, we have 
\begin{equation}
\label{eq:dimension-D-prime}
\dims{\codeDp} \leq m - \frac{3 i'}{4} - \frac{1}{4}\left(\frac{1-f}{3} + \frac{f}{1+g}\right)\xsubi{2}.
\end{equation}
\end{lemma}
\begin{proof}
Note that the pairs in $\App$ act as parity checks for the codewords of $\codeDp$. Thus, the support of any codeword $d\in\codeDp$ must be a union of connected components of $\Gamma$, otherwise it fails a parity check. Hence, $\dims{\codeDp}$ is at most the number of connected components of the subgraph of $\Gamma$ formed by the vertices in $\App \setminus (\Aonep\cup\Bonep)$. We denote such a restriction of $\Gamma$ to $\App \setminus (\Aonep\cup\Bonep)$ as $\codeproj{\Gamma}{\App \setminus (\Aonep\cup\Bonep)}$.

Now, note that every vertex of $\Gamma$ in $\App \setminus (\Aonep\cup\Bonep)$ has degree 3, and thus, the smallest possible connected component must be a complete graph on four vertices. Hence, the number of connected components in $\codeproj{\Gamma}{\App \setminus (\Aonep\cup\Bonep)}$ is at most $\frac{|\App \setminus (\Aonep\cup\Bonep)|}{4}$.

Thus, we have
\begin{IEEEeqnarray}{rCl}
\dims{\codeDp} & \leq & \frac{|\App \setminus (\Aonep\cup\Bonep)|}{4}\nonumber\\
& \leq & \frac{4m - 3i' - |\Aonep\cup\Bonep|}{4}\nonumber\\
& \leq & m - \frac{3 i'}{4} - \frac{1}{4}\left(|\Conep| + |\Aonep\setminus\Conep|\right)\\
& \leq & m - \frac{3 i'}{4} - \frac{1}{4}\left(\frac{(1-f)\xsubi{2}}{3} + \frac{f\xsubi{2}}{1+g}\right),
\label{eq:dimension-D-prime-2}
\end{IEEEeqnarray}
where~\eqref{eq:dimension-D-prime-2} follows from~\eqref{eq:C1-prime-size} and~\eqref{eq:A1-prime-minus-C1-prime-size}, and %$|\Conep| = \frac{(1-f)\xsubi{2}}{3}$ and $|\Aonep\setminus\Conep| \geq \frac{f\xsubi{2}}{2}$. 
the result follows from~\eqref{eq:dimension-D-prime-2}.
\end{proof}

\begin{lemma}
\label{lem:dimension-D-projection}
Denote by $\codeproj{\codeD}{\A \cup \Aonep \cup \Bonep}$ the projection of $\codeD$ on $\A \cup \Aonep \cup \Bonep$. Then, we have
\begin{equation}
\label{eq:dimension-D-projection}
\dims{\codeproj{\codeD}{\A \cup \Aonep \cup \Bonep}} \leq 2m + 2 i' - \frac{1}{3}\left(\frac{1-f}{3} + \frac{gf}{1+g}\right)\xsubi{2} - \frac{4}{9}\xsubi{3}.
\end{equation}
\end{lemma}
\begin{proof}

First, note that $\dims{\codeproj{\codeD}{\A \cup \Aonep \cup \Bonep}} = \dims{\codeproj{\codeD}{\A \cup \Conep}}$. Because, if the dimensions were different, then there would be a codeword in $\codeproj{\codeD}{\A \cup \Aonep \cup \Bonep}$ that vanishes on $\A \cup \Conep$, \ie, it is supported on the connected components that touch $\Aonep$. This codeword must then be the zero codeword.

%Next, to compute $\dims{\codeproj{\codeD}{\A \cup \Conep}}$, consider the set $\set{P}$ of parity checks in $\code$ that are orthogonal to the codewords of $\codeproj{\codeD}{\A \cup \Conep}$. Clearly, $\{S_1, \ldots, S_{\mip}, \Esubj{1}, \ldots, \Esubj{\xsubi{3}}\} \in \set{P}$. 

Now, for every point in $\Conep$, arbitrarily choose two of the three type $F$ subsets that cover the point, and add the subsets to obtain a parity check supported only on $\A$. Label such 4-subsets as $\Fsubj{1}', \ldots, \Fsubj{|\Conep|}'$.

Further, note that any vertex in $\Aonep\setminus\Conep$ with degree 1 is covered by two singletons. For each degree 1 vertex in $\Aonep\setminus\Conep$, add the two type $F$ subsets containing the two singletons to produce a parity check supported only on $\A$. Label such 4-subsets as $\Fsubj{|\Conep|+1}', \ldots, \Fsubj{|\Conep| + g|\Aonep\setminus\Conep|}'$. %Note that subsets $\{\Esubj{1}, \ldots, \Esubj{\xsubi{3}}, \Fsubj{1}', \ldots, \Fsubj{g|\Aonep\setminus\Conep|}'\}$ are the parity checks that vanish on $\Ap$. To find their rank, we define a graph on these subsets as follows.

Define a graph $\Gamma'$ with $\xsubi{3}$ type $E$ 3-subsets as blue vertices and $|\Conep| + g|\Aonep\setminus\Conep|$ type $F'$ 4-subsets as red vertices. Add $l$ edges between a pair of vertices if the corresponding subsets meet in $l$ points, where $1\leq l\leq 4$. 

%%%%%%%%%%%%%%%%%%%%%%%%%%%%%%
\begin{figure}[!t]
\centering 
\subfigure[Component of size 2]{
\includegraphics[scale=0.45]{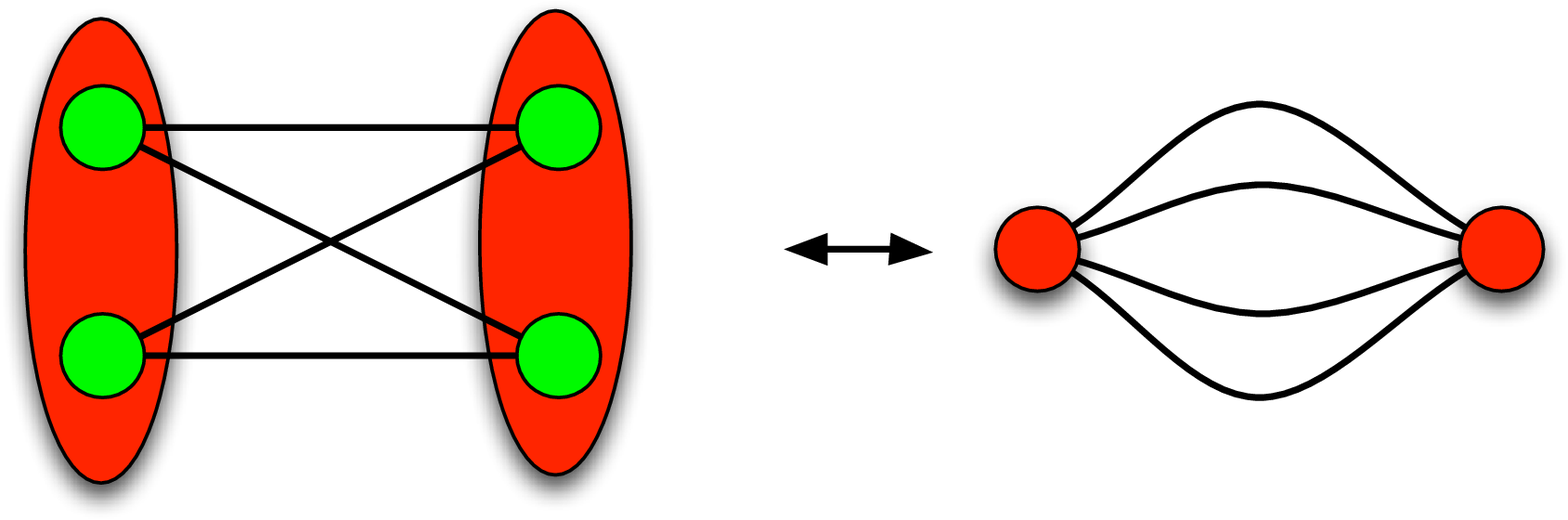}
}\qquad 
\subfigure[Component of size 3]{
\includegraphics[scale=0.45]{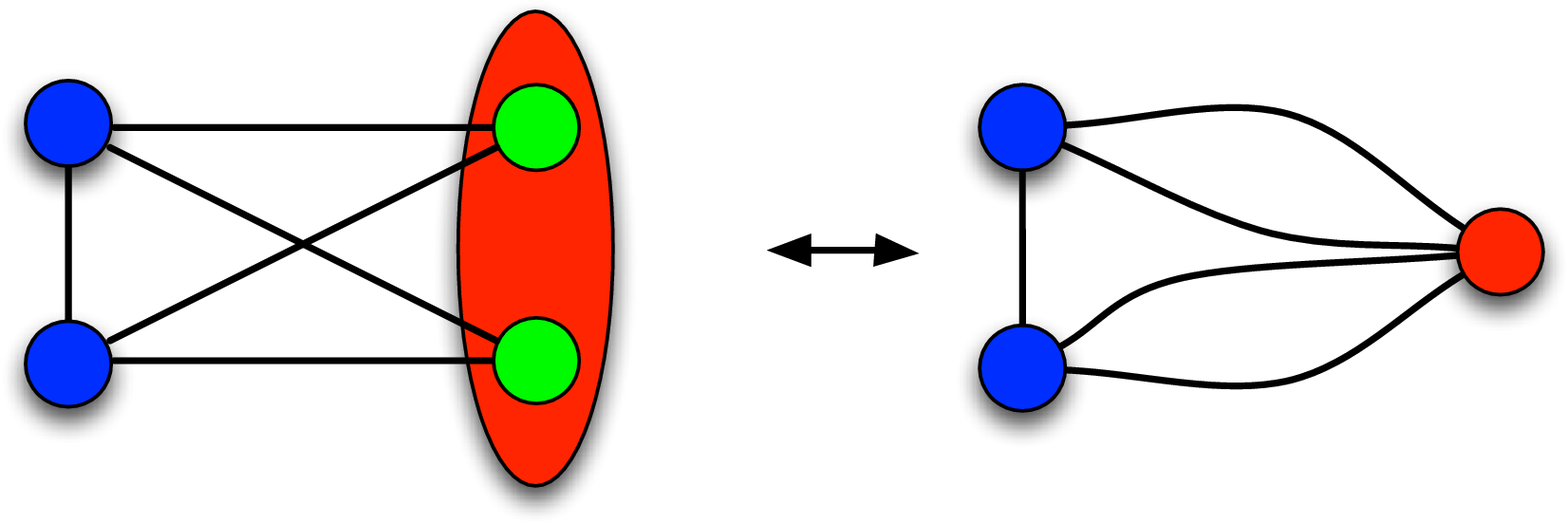}
}
\caption{Smallest possible connected components in $\Gamma'$ corresponding to a minimal linear dependency.} 
\label{fig:connected-components}
\end{figure}
%%%%%%%%%%%%%%%%%%%%%%%%%%%%%%

Note that we can view a red vertex as a {\it super-vertex} containing two disjoint green vertices, each corresponding to the pair of points in the type $F$ subset used to obtain a type $F'$ subset representing the red vertex. Further, note that the degree of a green vertex is at most 2, and thus, the degree of a red vertex is at most 4. On the other hand, the degree of a blue vertex is at most 3. 

For any (minimal) linear dependency $\sum_{i}\Esubj{i} + \sum_{j}\Fsubj{j}' = 0$, the blue vertices corresponding to $\Esubj{i}$'s and the red vertices corresponding to $\Fsubj{j}'$'s form a connected component in $\Gamma'$ such that every blue vertex has degree 3 and every red vertex has degree 4. Note that the smallest possible size of such a connected component containing all red vertices is 2, while the smallest possible size of such a connected component containing a blue vertex is 3. Fig.~\ref{fig:connected-components} depicts the smallest connected components.

Now, partition $\Gamma'$ into connected components, and eliminate one vertex from each connected component in which every blue vertex has degree 3 and every red vertex has degree 4. This yields at least $\frac{2}{3}\xsubi{3} + \frac{1}{2}\left(\frac{gf}{1+g}\xsubi{2} + \frac{1-f}{3}\xsubi{2}\right)$ vertices such that the corresponding vectors are linearly independent. 

Next, form a matrix $M$ with any $\frac{2}{3}\xsubi{3} + \frac{1}{2}\left(\frac{gf}{1+g}\xsubi{2} + \frac{1-f}{3}\xsubi{2}\right)$ linearly independent type $E$ and type $F'$ vectors, and reduce the matrix to row echelon form. Whenever there are three {\it diagonal} non-zero entries in $M$ that are are indexed by the same 3-subset $S_j$,  delete one of the three rows. Append the resulting matrix with the vectors $S_1, \ldots, S_{\mip}$. There cannot be any linear dependency in this matrix. Thus, we have % $\rnk{M} \geq m + i' + \frac{2}{3}\left(\frac{2}{3}\xsubi{3} + \frac{1-f}{6}\xsubi{3} + \frac{1}{2}\frac{fg}{1+g}\xsubi{2} \right)$. In other words, we have shown that
\begin{equation}
\label{eq:dim-D-proj-perp} 
\dims{\subspace{\Esubj{1}, \ldots, \Esubj{\xsubi{3}}, \Fsubj{1}', \ldots, \Fsubj{|\Conep|+g|\Aonep\setminus\Conep|}', S_1, \ldots, S_{\mip}}} \geq m + i' + \frac{2}{3}\left(\frac{2}{3}\xsubi{3} + \frac{1}{2}\left(\frac{gf}{1+g}\xsubi{2} + \frac{1-f}{3}\xsubi{2}\right)\right).
\end{equation}

Arbitrarily choose one of the three type $F$ subsets for every point in $\Conep$. Label them as $\Fsubj{1}, \ldots, \Fsubj{|\Conep|}$. None of them can be in the span of type $S$, type $E$ and type $F'$ subsets. Thus, we have
\begin{IEEEeqnarray}{C}
\label{eq:dim-D-proj-perp} 
\dims{\subspace{\Esubj{1}, \ldots, \Esubj{\xsubi{3}}, \Fsubj{1}', \ldots, \Fsubj{|\Conep|+|\Aonep\setminus\Conep|}', S_1, \ldots, S_{\mip}, \Fsubj{1}, \ldots, \Fsubj{|\Conep|} }}\nonumber\\ 
\geq |\Conep| + m + i' + \frac{2}{3}\left(\frac{2}{3}\xsubi{3} + \frac{1}{2}\left(\frac{gf}{1+g}\xsubi{2} + \frac{1-f}{3}\xsubi{2}\right)\right).
\end{IEEEeqnarray}
This allows us to write
\begin{equation}
\label{eq:dim-D-proj} 
\dims{\codeproj{\codeD}{\A \cup \Conep}} \leq  |\A \cup \Conep| - \left(|\Conep| + m + i' + \frac{2}{3}\left(\frac{2}{3}\xsubi{3} + \frac{1}{2}\left(\frac{gf}{1+g}\xsubi{2} + \frac{1-f}{3}\xsubi{2}\right)\right)\right),
\end{equation}
from which the result follows noting that $|\A| = 3m + 3i'$.
\end{proof}

Using Lemmas~\ref{lem:dimension-D-prime} and~\ref{lem:dimension-D-projection}, we get the following corollary.
\begin{corollary}
\label{cor:dimension-D}
We have $\dims{\codeD} \leq 3m$, with equality if and only if $i' = 0$.
\end{corollary}
\begin{proof}
From~\eqref{eq:rank-nullity-D},~\eqref{eq:dimension-D-prime} and~\eqref{eq:dimension-D-projection}, we get
\begin{equation}
\label{eq:dimension-D-bound}
\dims{\codeD} \leq 3m + \frac{5}{4}i' - \frac{1}{4}\left(\frac{1-f}{3} + \frac{f}{1+g}\right)\xsubi{2} - \frac{1}{3}\left(\frac{1-f}{3} + \frac{gf}{1+g}\right)\xsubi{2} - \frac{4}{9}\xsubi{3}.
\end{equation}

We want to show that 
\begin{equation}
\label{eq:i-prime-inequality}
\frac{5}{4}i' \leq \frac{1}{4}\left(\frac{1-f}{3} + \frac{f}{1+g}\right)\xsubi{2} + \left(\frac{1}{3}\frac{gf}{1+g} + \frac{1-f}{9}\right)\xsubi{2} + \frac{4}{9}\xsubi{3}.
\end{equation}
It is easy to check that the right hand side (RHS) above is an increasing function of $f$. We minimize the RHS by setting $f=0$, and, for contradiction, assume that 
\begin{equation}
\label{eq:i-prime-inequality-contradiction}
\frac{5}{4}i' > \left(\frac{1}{12} + \frac{1}{9}\right)\xsubi{2} + \frac{4}{9}\xsubi{3}.
\end{equation}

Now, since the number of points in type $E$, type $F$ and type $T$ subsets is $6m + 6i'$, we have
\begin{equation}
\label{eq:total-ones}
\xsubi{1} + 2 \xsubi{2} + 3 \xsubi{3} = 6m + 6i'.
\end{equation}
Further, as the total number of covering 3-subsets is $7m$, we have
\begin{equation}
\label{eq:total-subsets}
\xsubi{1} + \xsubi{2} + \xsubi{3} = 6m - i'.
\end{equation}
By subtracting~\eqref{eq:total-subsets} from~\eqref{eq:total-ones}, we get
$$\xsubi{2} + 2 \xsubi{3} = 7i',$$
which gives
\begin{equation}
\label{eq:i-prime}
i' = \frac{\xsubi{2} + 2 \xsubi{3}}{7}.
\end{equation}

From~\eqref{eq:i-prime-inequality-contradiction} and~\eqref{eq:i-prime},  we get
$$\frac{5}{4}\left(\frac{\xsubi{2} + 2 \xsubi{3}}{7}\right) > \frac{21}{108}\xsubi{2} + \frac{4}{9}\xsubi{3},$$
which is a contradiction. 

Hence,~\eqref{eq:i-prime-inequality} holds, and thus from~\eqref{eq:dimension-D-bound}, we have $\dims{\codeD} \leq 3m$. Further, the equality can happen if and only if $i = \xsubi{2} = \xsubi{3} = 0$.
\end{proof}
 
\subsubsection{Step 3}
\label{sec:step-3} 
\begin{lemma}
\label{lem:uniqueness}
If $\dims{\code} = 4m$, then $\code$ must be (equivalent to) a direct sum of the $m$ copies of the $[7,4]$ Hamming code. 
\end{lemma}
\begin{proof}
First note that from Corollary~\ref{cor:dimension-D}, it follows that if $\dims{\code} = 4m$, then the size of a maximum collection of pairwise disjoint covering $3$-subsets in $\code$ is exactly $m$.
Next, we prove the result by induction on $m$.

{\it Basis Step:} $m = 1$. Since no two $3$-subsets can be disjoint, every pair of $3$-subsets must intersect. Thus, the $7$ $3$-subsets correspond to the Fano plane. The result follows since the row space of any incidence matrix of the Fano plane is isomorphic to the $[7,4]$ Hamming code~\cite{Assmus:92}.

{\it Induction Step:} $m\geq 2$. Consider a maximum collection of pairwise disjoint $3$-subsets of size $m$ as $\{S_1, \cdots, S_m\}$. Let $L$ be the subset of all $3$-subsets that are disjoint from $\{S_1, \cdots, S_{m-1}\}$. Due to exact covering, each $3$-subset intersects six other $3$-subsets, and thus, we have $|L| \geq 7$. Since $S_m \in T$, and there are $6$ other $3$-subsets that intersect $S_m$, we have $|L| = 7$. As there are no $m+1$ pairwise disjoint $3$-subsets, the $3$-subsets in $L$ must intersect pairwise.

Now, pick any subset $T\in L$. The six $3$-subsets that intersect $T$ must be the six other $3$-subsets in $L$. Thus, any $3$-subset in $L$ must be disjoint from any $3$-subset outside $L$. Further, the $3$-subsets in $L$ must cover $7$ points due to the availability of the points. 

Let $\code_1$ denote the restriction of $\code$ on the points covered by the $3$-subsets in $L$, and $\code_2$ denote the restriction of $\code$ on the points covered by the $3$-subsets outside $L$. Then, we have $\code = \code_1 \oplus \code_2$. Also, since the $3$-subsets in $L$ pairwise intersect, they correspond to the Fano plane and $\code_1$ must be equivalent to the  $[7,4]$ Hamming code. In addition, as $\dims{\code} = 4m$, it must be that $\dims{\code_2}$ is a $[7(m-1), 4(m-1)]$ code. Thus, the result follows by induction. 
\end{proof}

\section{Rate Upper Bounds Using Coset Leaders} %Codes with $(2,t)$ and $(r,3)$-Availability
\label{sec:bounds}

\subsection{Rate Bound for Codes with $(2,t)$-Availability}
First, we present a bound on the rate of a binary code having $(2,t)$-availability with exact covering. Our main idea is to bound the maximum weight of a coset leader of its dual code by using the covering properties imposed by availability constraints. We note that the maximum weight of a coset leader of a linear code represents its {\it covering radius}~\cite{Cohen:85}.

\begin{theorem}
\label{thm:rate-bound-2-t}
Let $\code$ be a length-$n$ code spanned by $\frac{nt}{3}$ $3$-subsets that cover every point with $(2,t)$-availability. Then, we have %For any binary, linear code $\codeD$ with strict $(2,t)$-availability, we have
\begin{equation}
\label{eq:rate-bound-2-t}
\rate{\dual} \leq H_2\left(\frac{1}{t+1}\right),
\end{equation}
where $H_2(\cdot)$ is the binary entropy function.
\end{theorem}
\begin{proof}
%Let $\code$ be a code that covers every point with strict $(2,t)$-availability. Note that $\code$ is a dual of $\codeD$, and it is spanned by $tn/3$ 3-subsets, referred to as triples, such that any pair of triples intersect in at most one point. 
We refer to the $nt/3$ covering $3$-subsets as {\it triples}. 
Let $\vl$ be a coset leader of a coset of $\code$ such that $\wt{\vl} = \w$. By the minimality of $\w$, every triple should meet $\vl$ in at most one point. 

Let $\Gamma$ be a graph formed on the complement of $\vl$ by the $\w t$ triples that meet $\vl$ in one point, defined as follows. Vertices of $\Gamma$ are the $n-\w$ points in the complement of $\vl$, and a pair of vertices are connected by an edge if the corresponding points belong to a triple. Note that the number of edges in $\Gamma$ is $\w t$, whereas the number of vertices in  $\Gamma$ is $n - \w$. 

Our main goal is to show that $\w \leq n/(t+1)$. Towards this end, we note the following properties of $\Gamma$. First, $\Gamma$ does not contain any cycle of odd length. This is because if $\Gamma$ contains a cycle of odd length, then the sum of corresponding triples is a codeword of odd weight supported within $\vl$. This contradicts the assumption that $\vl$ is a coset leader. 

Second, the number of edges in $\Gamma$ is at most the number of vertices in it.
If the maximum degree in $\Gamma$ is two, then the result follows. Otherwise, let $v_0$ be a vertex in  $\Gamma$ of degree greater than two. Then, in the following, we show that any neighbor of $v_0$ cannot have degree greater than one.

Let $P$, $Q$, and $R$ be any three triples intersecting in the point corresponding to $v_0$. Denote the points in $P$~(respectively, $Q$ and $R$)~as $P_i$ (respectively, $Q_i$ and $R_i$) for $i = \{1, 2, 3\}$. Let $P_1$, $Q_1$, and $R_1$ be the points that meet $\vl$. Let $P_2$, $Q_2$, and $R_2$ correspond to the vertex $v_0$. Note that $P_3$, $Q_3$, and $R_3$ correspond to the neighbors of $v_0$.

Suppose, for contradiction, that $P_3$ corresponds to a vertex of degree two or more. Let $S$ be a triple meeting $\vl$ that intersects $P$ in $P_3$. Note that $S$ cannot contain $P_3$ or $Q_3$, as this would result in a triangle (which is an odd cycle) in $\Gamma$. Let $\vl' = P + Z + S$, where $Z$ is chosen to be either $Q$ or $R$ such that it is disjoint from $S$. Then, we have $\wt{\vl + \vl'} < \wt{\vl}$, which contradicts that $\vl$ is a coset leader. Thus, the vertex corresponding to $P_3$ cannot have degree greater than one. This proves that every neighbor of a vertex of $\Gamma$ of degree greater than two must have degree one. In other words, $\Gamma$ consists of (even length) cycles, paths, and stars. Hence, the number of edges in $\Gamma$ is at most the number of vertices. This yields that $\w \leq n/(t+1)$.

Now, we use the bound on $\w$ to limit the number of cosets of $\code$, which allows us to lower bound the dimension of $\code$ as follows. Let $\w_{max}$ denote the maximum weight of a coset leader of $\code$. Then, we can write 
\begin{IEEEeqnarray}{rCl}
\dim{\code} & = & \log_2\left(\frac{2^n}{\textrm{Number \: of \: cosets \: of \:} \code}\right),\nonumber\\ 
& \geq & \log_2\left(\frac{2^n}{\sum_{i=0}^{\w_{max}}\binom{n}{i}}\right),\nonumber\\
& \geq & \log_2\left(\frac{2^n}{2^{n H_2\left(\frac{\w_{max}}{n}\right)}}\right),\nonumber\\
& = & n\left( 1 - H_2\left(\frac{\w_{max}}{n}\right)\right),\nonumber\\
& \geq & n\left( 1 - H_2\left(\frac{1}{t+1}\right)\right),
\label{eq:dim-C-lower-bound}
\end{IEEEeqnarray}
where the second inequality follows from the well-known result that $\sum_{i=0}^{\w_{max}}\binom{n}{i} \leq 2^{n H_2\left(\frac{\w_{max}}{n}\right)}$; and the last inequality holds because $\w_{max} \leq n/(t+1)$ and $H_2(x)$ is increasing in $x$ for $0\leq x\leq 1/2$. 

The bound in~\eqref{eq:rate-bound-2-t} follows from~\eqref{eq:dim-C-lower-bound}.
\end{proof}

{\bf Tight Rate Bound for Length-$n$ Codes with $\left(2,\frac{n-1}{2}\right)$-Availability and Optimality of Simplex Codes:}
Using the idea of bounding the weight of a coset leader, we can easily obtain a tight upper bound on the rate of codes with $\left(2,\frac{n-1}{2}\right)$-availability. As we will see, when $n = 2^m - 1$ for a positive integer $m$, this bound is achieved by the $(2^m - 1, m, 2^{m-1})$ Simplex code.  

\begin{theorem}
\label{thm:rate-bound-Simplex}
Let $\code$ be a length-$n$ code spanned by $\frac{n(n-1)}{6}$ $3$-subsets that cover every point with $\left(2,\frac{n-1}{2}\right)$-availability. Then, we have %For any binary, linear code $\codeD$ with strict $(2,t)$-availability, we 
\begin{equation}
\label{eq:rate-bound-Simplex}
\rate{\dual} \leq \frac{\log_2(n + 1)}{n}.
\end{equation}
\end{theorem}
\begin{proof}
The proof follows from the observation that the weight of a coset leader of $\code$ should be at most one. This is because every triple must intersect a coset leader in at most one point due to the minimality of its weight.  
\end{proof}

\begin{remark}
It has been observed that the $(2^m - 1, m, 2^{m-1})$ Simplex code has $(2,2^{m-1})$-availability, see, \eg~\cite{Goparaju:14, Wang:15ISIT}. The $(2^m - 1, m, 2^{m-1})$ Simplex code achieves the bound in~\eqref{eq:rate-bound-Simplex}.
We note that the rate optimality of Simplex codes amongst binary codes with locality $r = 2$ has been shown in~\cite{CadambeM:15} using their field size dependent bound. The idea of bounding the weight of a coset leader gives a very simple proof for this result.
\end{remark} 

\subsection{Bound for Codes with $(r,3)$-Availability}
The bound in Theorem~\ref{thm:rate-bound-2-t} enables us to obtain, as a corollary, a rate upper bound for binary codes having $(r,3)$-availability with exact covering. The main idea is a simple yet powerful observation from~\cite{Balaji:16bounds}, stated in the following remark. 

\begin{remark}
\label{rem:BK-observation}
Let $H$ be a parity-check matrix of an $(n,k)$ code having $(r,t)$-availability with exact covering. Then, its transpose $H^T$ is a parity-check matrix for an $\left(\frac{nt}{r+1},k\right)$ code having $(t-1,r+1)$-availability with exact covering. 
\end{remark}

\begin{corollary}
\label{cor:rate-bound-r-3}
%For any binary, linear code $\codeD$ with strict $(r,3)$-availability, we have
Let $\code$ be a length-$n$ code spanned by $\frac{3n}{r+1}$ $(r+1)$-subsets that cover every point with $(r,3)$-availability. Then, we have
\begin{equation}
\label{eq:rate-bound-r-3}
\rate{\dual} \leq \frac{r-2}{r+1} + \frac{3}{r+1}H_2\left(\frac{1}{r+2}\right),
\end{equation}
where $H_2(\cdot)$ is the binary entropy function.
\end{corollary}
\begin{proof}
%If $H$ is a parity-check matrix of an $(n,k)$ code having $(r,t)$-availability with exact covering. Let $H$ be a parity-check matrix of $\code$. Then, its transpose $H^T$ is a parity-check matrix for an $\left(\frac{3n}{r+1},k\right)$ code having $(2,r+1)$-availability with exact covering. 
%Now, %let $\code$ be code covering every point with $(r,3)$-availability. 
Using Remark~\ref{rem:BK-observation} and~\eqref{eq:dim-C-lower-bound}, we get
\begin{equation}
\label{eq:dim-C-lower-bound-r-3}
\dims{\code} \geq \frac{3n}{r+1}\left[1 - H_2\left(\frac{1}{r+2}\right)\right],
\end{equation}
from which the result follows.
\end{proof}

{\bf Tight Rate Bound for Codes with $\left(r,3\right)$-Availability and Length $\frac{(r+1)(2r+3)}{3}$:} We get the following rate bound using Theorem~\ref{thm:rate-bound-Simplex} and Remark~\ref{rem:BK-observation}. %the observation that, for a code with $(r,t)$-availability, the transpose of its parity-check matrix is a parity-check matrix of a code with $(t-1,r+1)$-availability. 

\begin{corollary}
\label{cor:rate-bound-r-3-Simplex}
%For any binary, linear code $\codeD$ with strict $(r,3)$-availability, we have
Let $r$ be a positive integer such that $3$ is a divisor of $(r+1)(2r+3)$. Let $\code$ be a code with length $\frac{(r+1)(2r+3)}{3}$ and $(r,3)$-availability with exact covering. Then, we have
\begin{equation}
\label{eq:rate-bound-r-3-Simplex}
\rate{\dual} \leq 1 - \frac{3}{r+1} + \frac{3\log_2(2r+4)}{(r+1)(2r+3)}.
\end{equation}
\end{corollary}

\begin{remark}
\label{rem:}
Consider a $\left(2^m - 1, m, 2^{m-1}\right)$ Simplex code. Due to its $(2,2^{m-1})$-availability with exact covering (see~\cite{Goparaju:14, Wang:15ISIT}), it has a $\frac{(2^m-1)(2^{m-1}-1)}{3} \times (2^m - 1)$ parity-check matrix $H$ with column weight $3$ and row weight $2^{m-1}-1$ such that any pair of rows intersecting in at most one point. The code with $H^T$ as its parity-check matrix has $(2^{m-1}- 2, 3)$-availability, and it achieves the bound in~\eqref{eq:rate-bound-r-3-Simplex}.
\end{remark}

\subsection{Comparison with the Existing Bounds}
\label{sec:rate-bounds-comparison}
%{\bf Comparison with the Existing Bounds:}
We compare our bounds with~\eqref{eq:Tamo-Barg-rate-bound} from~\cite{Tamo:14Availability, Tamo:16bounds}, referred to as {\it TBF bound $1$}.  The authors of~\cite{Tamo:14Availability, Tamo:16bounds} also show that the expression on the right hand side of~\eqref{eq:Tamo-Barg-rate-bound} can be upper bounded by $\frac{1}{\sqrt[r]{t+1}}$, referred to as {\it TBF bound 2}.  
%The authors of~\cite{TamoB:14} give an upper bound on the rate of a code $\code$ with $(r,t)$-availability as follows.
%\begin{IEEEeqnarray}{rCl}
%\label{eq:TBF-bound-1}
%\rate{\code} & \leq & \frac{1}{\prod_{j=1}^{t}\left(1 + \frac{1}{jr}\right)}\\
%\label{eq:TBF-bound-2} 
%&\leq& \frac{1}{\sqrt[r]{t+1}}.
%\end{IEEEeqnarray}

We also compare our bound in~\eqref{eq:rate-bound-2-t} with the following bound on the rate of a code $\code$ with $(r,t)$-availability given in~\cite{Balaji:16bounds}. % present a rate upper bound for a code $\code$ with $(r,t)$-availability as
\begin{equation}
\label{eq:BK-bound-r-t}
\rate{\code} \leq 1 - \frac{t}{r+1} + \frac{t}{r+1} \frac{1}{\prod_{j=1}^{r+1}\left(1 + \frac{1}{j(t-1)} \right)}.
\end{equation}
We refer to~\eqref{eq:BK-bound-r-t} as {\it BK bound 1}.

Our bound in~\eqref{eq:rate-bound-2-t} is plotted as a function of $t$ in Fig.~\ref{fig:bounds-t-2}, along with  %bounds~\eqref{eq:Tamo-Barg-rate-bound} (referred to as Tamo-Barg 1), \eqref{eq:TBF-bound-2} (referred to as Tamo-Barg 2), and \eqref{eq:BK-bound-r-t} (referred to as Balaji-Kumar) 
TBF bound 1, TBF bound 2, and BK bound 1 for $r=2$. Observe that our bound gets sharper as $t$ increases crossing TBF bound 1 at $t = 74$. This advantage is clarified in Fig.~\ref{fig:bounds-t-2}, which zooms into the range $t = 35$ to $100$ in Fig.~\ref{fig:bounds-t-2-magnified}. 

%%%%%%%%%%%%%%%%%%%%%%%%%%%%%%
\begin{figure}[!t]
\centering
\includegraphics[scale=0.40]{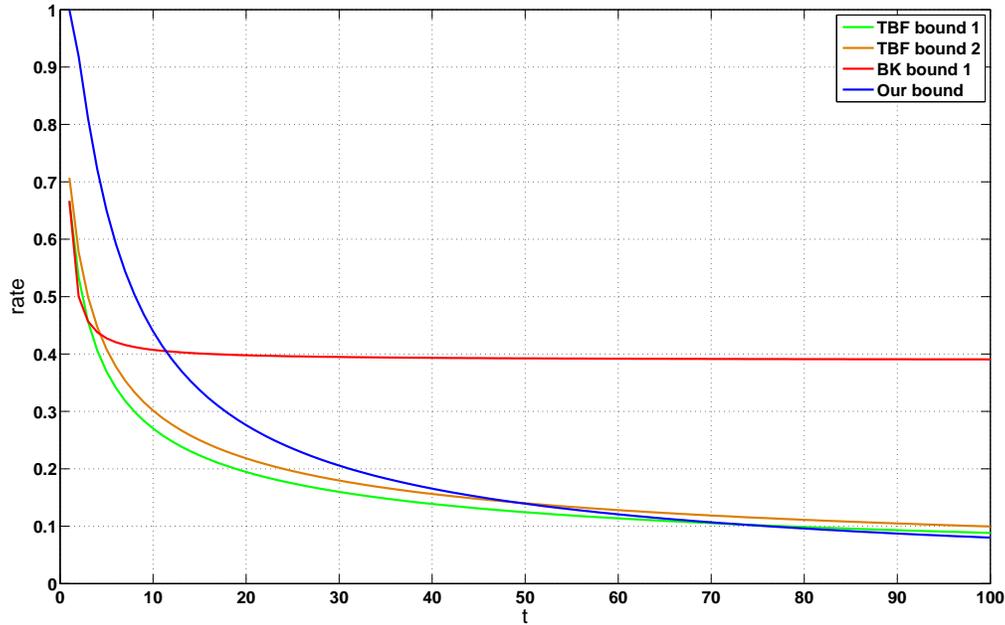}
\caption{Rate upper bounds versus $t$ for $r=2$.} 
\label{fig:bounds-t-2}
\end{figure}
%%%%%%%%%%%%%%%%%%%%%%%%%%%%%%

%%%%%%%%%%%%%%%%%%%%%%%%%%%%%%
\begin{figure}[!t]
\centering
\includegraphics[scale=0.40]{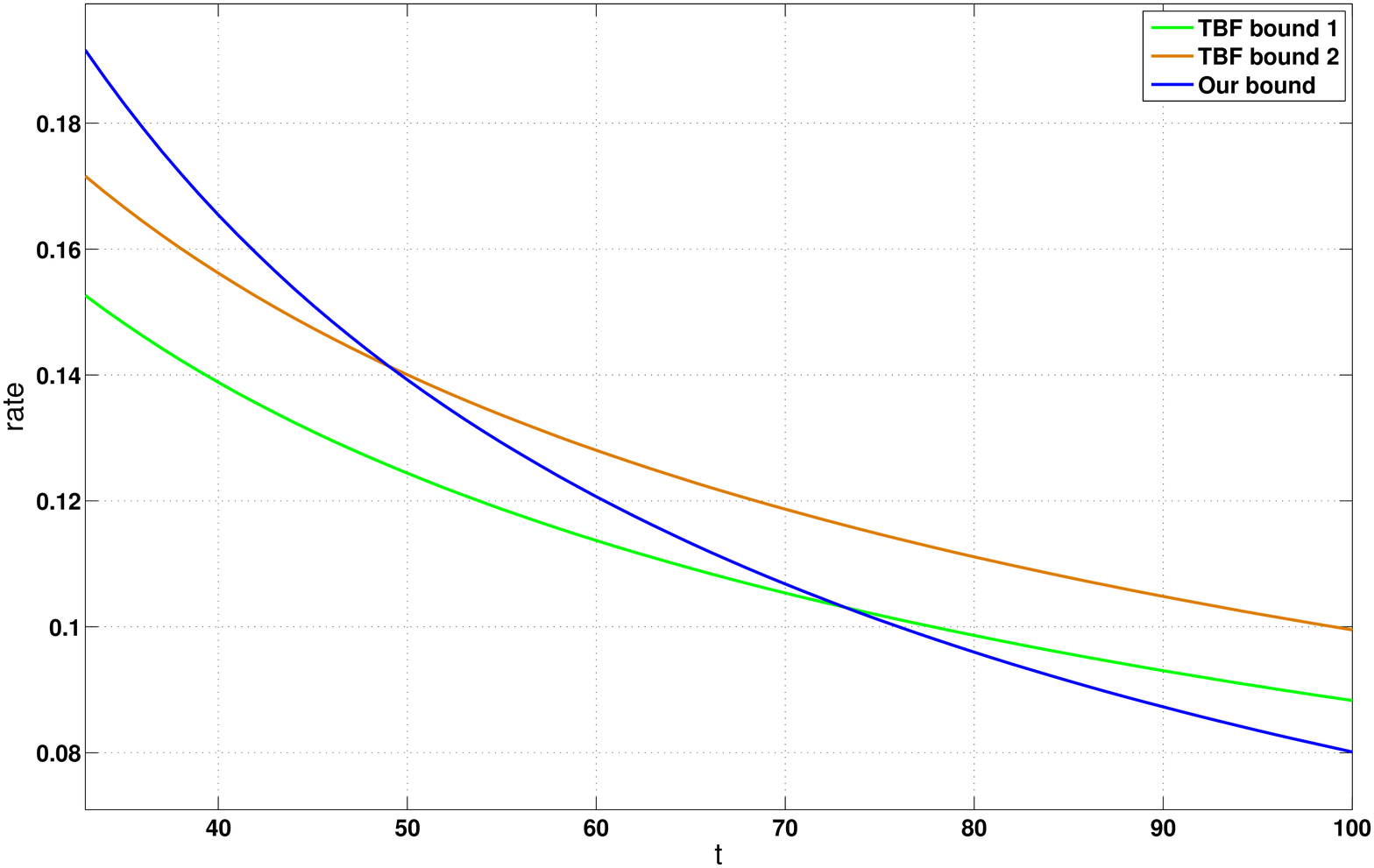}
\caption{Magnified version of the rate upper bounds for a range of $t$ when $r=2$.} 
\label{fig:bounds-t-2-magnified}
\end{figure}
%%%%%%%%%%%%%%%%%%%%%%%%%%%%%%

Next, we compare our bound in~\eqref{eq:rate-bound-r-3} with TBF bound 1, TBF bound 2, BK bound 1, and the following bound from~\cite{Balaji:16bounds} on the rate of a code $\code$ with $(r,3)$-availability. %specialized for $t = 3$. %present a rate upper bound for an $(n,k)$ code $\codeD$ with $(r,3)$-availability as
\begin{equation}
\label{eq:BK-bound-2}
\rate{\code} \leq 1  - \frac{3(1 + L_1 + L_2)}{(r+1)(3+L_1+2L_2)},
\end{equation}
where $m = \frac{3n}{r+1}$, $L'_1 = \left\lceil\frac{(2r-1)m}{3(r+2)} - \frac{1}{r+1} - 1\right\rceil$, $L_2 = \left\lfloor\frac{m - 3 - L'_1}{2}\right\rfloor$, and
$L_1 = m - 3 - 2L_2$. We refer to~\eqref{eq:BK-bound-2} as {\it BK bound 2}.

We plot our bound in~\eqref{eq:rate-bound-r-3} as a function of $r$ in Fig.~\ref{fig:bounds-r-3}, along with TBF bound 1, TBF bound 2, BK bound 1 for $t=3$, and BK bound 2 for $n = \binom{r+3}{3}$. Our bound is loose for small values of $r$, but it gets sharper as $r$ increases, crossing BK bound 1 at $r = 72$. The gap with BK bound 1 is very small, on the order of $10^{-4}$, which we clarify in Fig.~\ref{fig:bounds-r-3} by zooming into the range $r = 40$ to $90$ in Fig.~\ref{fig:bounds-r-3-magnified}. Note that the block-length $n$ appears explicitly in the expression of BK bound 2 in~\eqref{eq:BK-bound-2}. We observed the same trend as shown in Fig.~\ref{fig:bounds-r-3} for different values of $n$, which we do not include for the want of space.

%%%%%%%%%%%%%%%%%%%%%%%%%%%%%%
\begin{figure}[!t]
\centering
\includegraphics[scale=0.40]{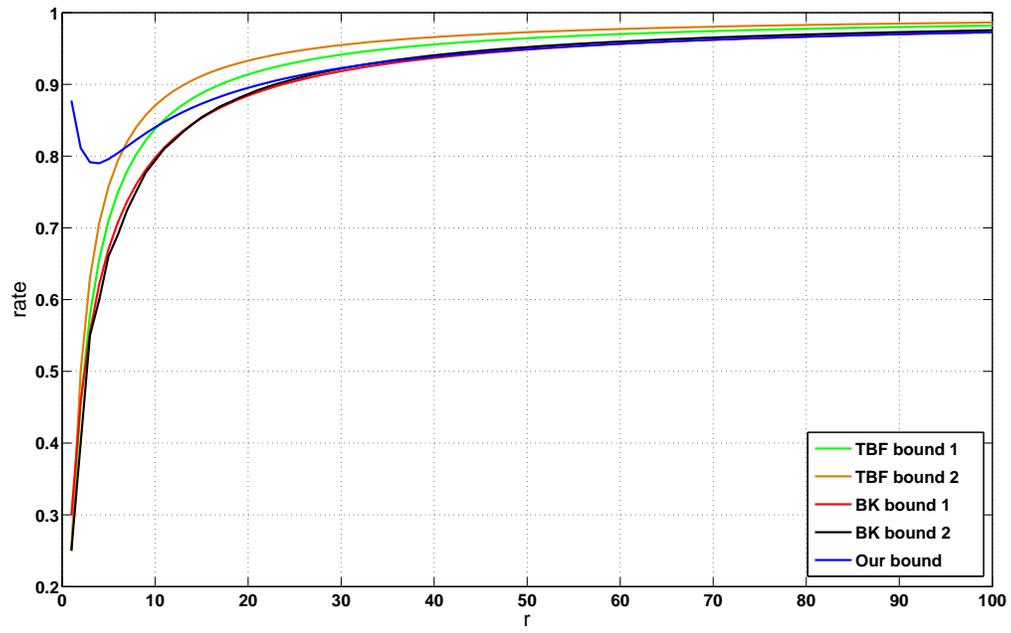}
\caption{Rate upper bounds versus $r$ for $t=3$.} 
\label{fig:bounds-r-3}
\end{figure}
%%%%%%%%%%%%%%%%%%%%%%%%%%%%%%

%$$ \left(S_1\right) \cdots \left(S_{\mip}\right) $$
%$$ \left(\Esubj{1}\right) \cdots \left(\Esubj{\xsubi{3}}\right) $$
%$$ \left(\Fsubj{1}\right) \cdots \left(\Fsubj{\xsubi{2}}\right) $$
%$$ \left(\Tsubj{1}\right) \cdots \left(\Tsubj{\xsubi{1}}\right) $$
%$$ \A \quad \Ap \quad \Aonep \quad \Bonep \quad \Conep \quad \App$$

%\end{proof}

%%%%%%%%%%%%%%%%%%%%%%%%%%%%%%
\begin{figure}[!t]
\centering
\includegraphics[scale=0.40]{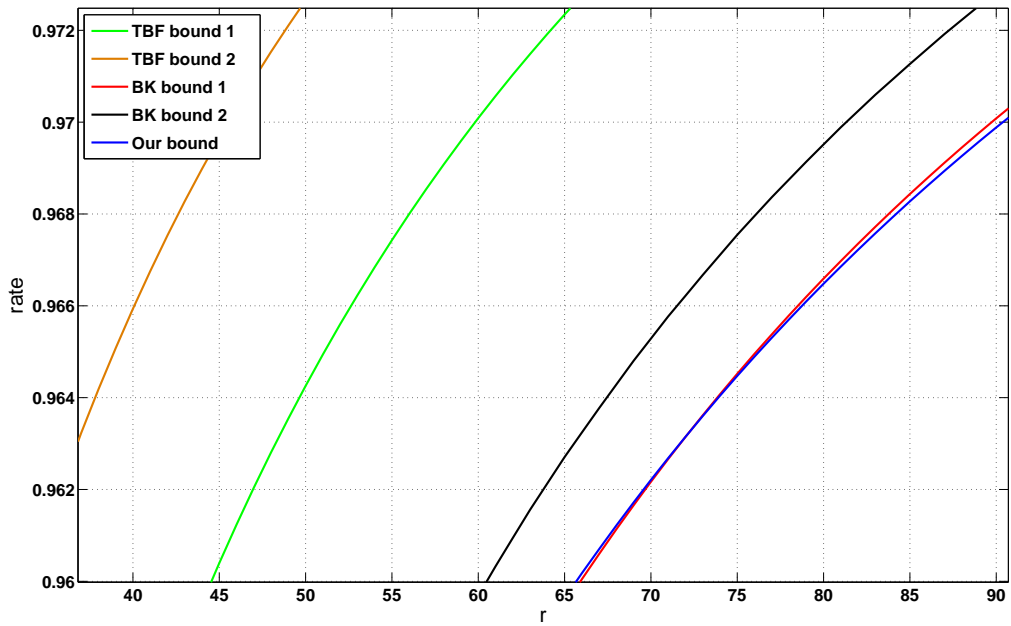}
\caption{Magnified version of the rate upper bounds for a range of $r$ when $t=3$.} 
\label{fig:bounds-r-3-magnified}
\end{figure}
%%%%%%%%%%%%%%%%%%%%%%%%%%%%%%

\section{Concluding Remarks}
\label{sec:discussion}
We studied availability properties of codes associated with convex polyhedra, focusing on the codes associated with the Platonic solids. Further, we computed tight upper bounds on the rate of binary linear codes with $(r,2)$ and $(2,3)$-availability, and showed the uniqueness of direct sum type constructions for rate optimality. Our main idea is to view the problem of designing a rate-optimal code with $(r,t)$-availability as a {\it covering problem}. Since direct sum constructions are known to give good codes for conventional covering problems~\cite{Cohen:85}, we speculate that such a direct sum construction will be present in rate-optimal codes for other values of $r$ and $t$. Finally, we presented novel upper bounds on the rates of binary linear codes with $(2,t)$ and $(r,3)$-availability.
%We speculate that such a direct sum construction will be present in rate optimal codes for other values of r and t

\section*{Acknowledgment}
S. Kadhe would like to thank Anoosheh Heidarzadeh, Krishna Narayanan, and Alex Sprintson for helpful discussions.

%\nocite{*}
\bibliographystyle{IEEEtran}
\bibliography{Bib_availability_v1}

\end{document}